%% file: paper.tex
\documentclass{article}
\usepackage{pacc_preprint,times}
\usepackage[T1]{fontenc}
\usepackage{microtype}
\input{math_commands.tex}

\usepackage{subcaption}
\usepackage{graphicx}
\usepackage{multirow}
\usepackage{makecell}
\usepackage{array}
\usepackage{booktabs}
\usepackage{longtable}
\usepackage{ragged2e}
\usepackage{amsmath}
\usepackage{amsthm}
\usepackage{xspace}
\usepackage{enumitem}
\usepackage{thmtools}
\usepackage{thm-restate}
\usepackage{xcolor}

\usepackage{amssymb}
\usepackage{pifont}
\usepackage{placeins}
\usepackage{stfloats}
\usepackage{float}
\usepackage{hyperref}
\hypersetup{
    pdftitle={PACC: Protocol-Aware Cross-Layer Compression for Compact Network Traffic Representation},
    pdfauthor={Zhaochen Guo; Honghao Wang; Tianyufei Zhou; Ronghua Li; Shinan Liu},
    colorlinks=true,
    citecolor=blue,
    linkcolor=blue,
    urlcolor=blue
}
\declaretheorem[name=Definition, numberwithin=section]{myDef}

\newcommand{\sysname}{\texttt{PACC}\xspace}

\makeatletter

\newcommand{\tightsubsection}[1]{%
  \@startsection{subsection}{2}{\z@}%
    {-0.8ex plus -0.2ex minus -.1ex}%
    {0.3ex plus .1ex}%
    {\normalsize\sc\raggedright}%
    {#1}%
}

\newcommand{\tightsection}[1]{%
  \@startsection{section}{1}{\z@}%
    {-0.6ex plus -0.2ex minus -.1ex}%
    {0.8ex plus .1ex}%
    {\large\sc\raggedright}%
    {#1}%
}

\makeatother

\begin{document}

\title{\raggedright\hyphenpenalty=10000 \sysname: Protocol-Aware Cross-Layer Compression for Compact Network Traffic Representation}
\author{%
\textbf{Zhaochen Guo\textsuperscript{1,2} \quad Honghao Wang\textsuperscript{3} \quad Tianyufei Zhou\textsuperscript{1}}\\[2pt]
\textbf{Ronghua Li\textsuperscript{1} \quad Shinan Liu\textsuperscript{1}}\\[5pt]
{\normalfont\small\textsuperscript{1}\,University of Hong Kong \quad
\textsuperscript{2}\,University of Electronic Science and Technology of China}\\
{\normalfont\small\textsuperscript{3}\,Renmin University of China}}
\date{}
\maketitle

\begin{abstract}
\looseness=-1
Network traffic classification underpins network security and management but faces growing challenges from pervasive encryption and evolving protocols. Existing representations involve clear trade-offs: hand-crafted flow statistics are efficient but lossy, raw-bit encodings can be accurate but costly, and pre-trained embeddings support transfer but often flatten the protocol stack and entangle cross-layer signals.
We observe substantial \emph{redundancy} both \emph{across} and \emph{within} network layers, which existing paradigms do not explicitly address, wasting capacity and encouraging shortcut learning that degrades generalization. To address this, we propose \sysname, a redundancy-aware cross-layer representation framework. \sysname models protocol layers as multiview inputs and learns faithful, compact projections factorized into \emph{shared} (cross-layer) and \emph{private} (layer-specific) components.
It incorporates an information-preserving sparse segment representation to reduce structural repetition and supervised intra-layer filtering to suppress task-irrelevant semantic redundancy. A behavior-guided shared/private decomposition aligns cross-layer consensus while retaining layer-specific information, supported by information-theoretic analysis.
Across encrypted application classification, IoT device identification, benign/malicious traffic classification, and Tor traffic detection, \sysname achieves the best overall performance. On encrypted traffic, it achieves an F1 score of \(0.9685\), substantially outperforming existing baselines while using approximately \(15\times\) fewer computes and less GPU memory than the latest pre-trained traffic model.

\end{abstract}

\input{sections/introduction.tex}

\input{sections/related}
\input{sections/preliminaries}
\input{sections/methods}

\input{sections/experiment}

\input{sections/conclusion}
\input{sections/statement}

\bibliography{reference}
\bibliographystyle{pacc_preprint}

\appendix
\input{sections/appendix}

\end{document}

%% file: math_commands.tex
\usepackage{amsmath,amsfonts,bm}

\def\eqref#1{equation~\ref{#1}}

\def\1{\bm{1}}

\DeclareMathAlphabet{\mathsfit}{\encodingdefault}{\sfdefault}{m}{sl}
\SetMathAlphabet{\mathsfit}{bold}{\encodingdefault}{\sfdefault}{bx}{n}

%% file: sections/introduction.tex
\section{Introduction}

Network traffic classification underpins network security
(e.g., intrusion detection~\citep{mirsky2018kitsuneensembleautoencodersonline,gupta2025generative}, malware analysis~\citep{shibahara2016efficient}, anomaly detection~\citep{radford2018network,li2025wifinger}),  network management
(e.g., application identification~\citep{jiang2025jiti,shaowang2025algorithmic,chu2025netssm}, QoE inference~\citep{qoe2025kulkarni,liu2024serveflow,mangla2018emimic} and capacity planning~\citep{liu2023leaf,zhu2021network}).
As most traffic payload is now encrypted~\citep{google2026https}, classification must infer application and security semantics from partially observable traces, with useful evidence distributed across protocol layers.

Existing representations range from coarse-grained flow statistics \citep{jiang2025jiti,liu2023amir} and packet-sequence features\citep{8737507} to raw-bit/byte encodings\citep{tfegnn2023zhang} and pre-trained embeddings\citep{lin2022bert}. Although effective in different settings, they may face a granularity-cost tradeoffs, or entangle signals across protocol layers. Fundamentally, they are not designed to reflect a key insight built upon the layered principles of the Internet:
\emph{network traffic exhibits substantial redundancy both across layers and within layers}.

Based on this insight, we distinguish redundancy by its role.
\emph{Intra-layer redundancy} includes repetitive structure (e.g., recurring header fields, header options padding, and predictable field states), and task-irrelevant semantic noise (e.g., session-specific artifacts). These inflate representations and consume computation without necessarily improving classification.
In contrast, \emph{inter-layer redundancy} expresses related semantics across layers, forming cross-layer consensus that can provide robust evidence.
Each layer may also contain \emph{non-redundant} task-relevant information unexplained by this consensus.
Our goal is therefore to compress repetitive intra-layer structure, suppress task-irrelevant semantic noise, align inter-layer shared evidence, and preserve complementary layer-specific information.
Fig.~\ref{fig:venn} illustrates this shared--private principle through mutual information, $I(\cdot;\cdot)$.

\begin{figure}[t]
    \centering
    \includegraphics[width=0.75\linewidth]{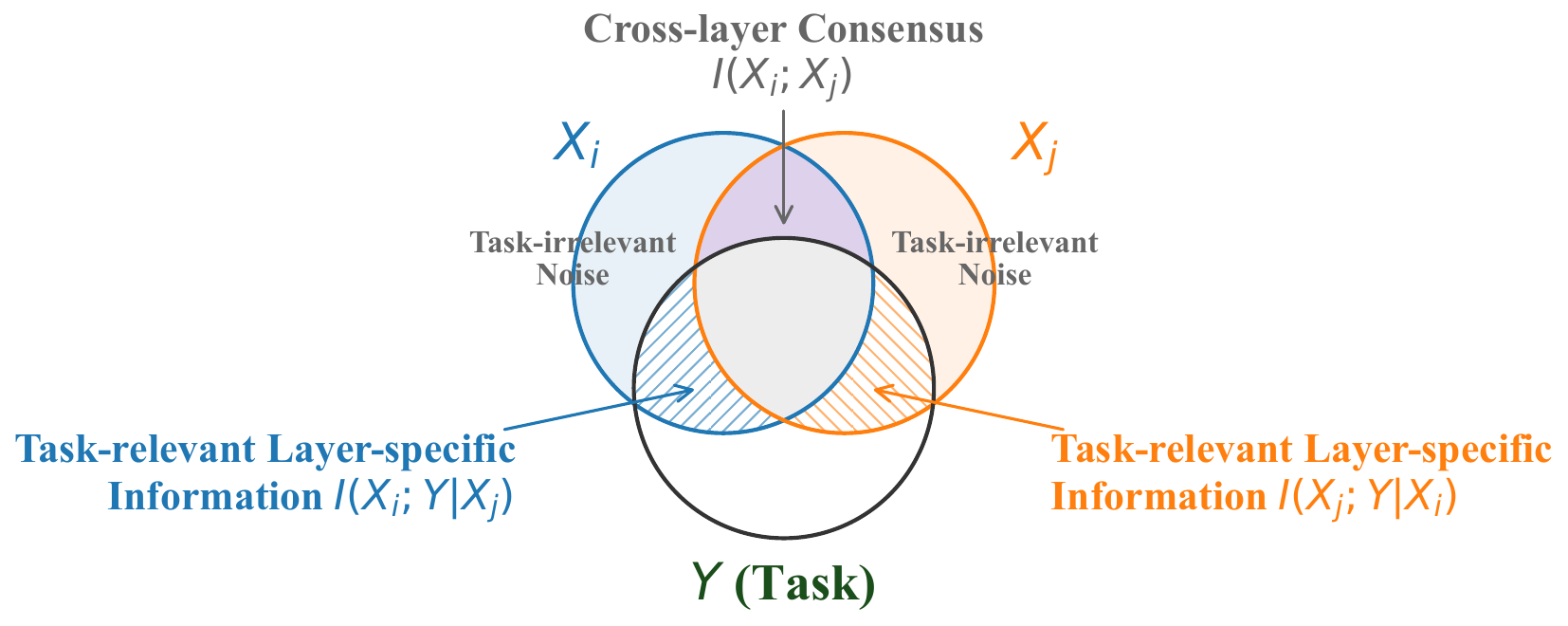}
    \vspace{-0.4cm}
    \caption{
    The shared--private principle for cross-layer traffic: capturing cross-layer consensus $I(X_i;X_j)$ and layer-specific, task-relevant evidence such as $I(X_i;Y \mid X_j)$.}
    \label{fig:venn}
    \vspace{-0.8cm}
\end{figure}

To examine this motivation empirically, we analyze raw layer representations on CipherSpectrum~\citep{11023502}, whose application payload is encrypted.
We measure inter-layer consensus using $\mathrm{NMI}(X_i;X_j)$ between discretized layer representations and task relevance using $\mathrm{NMI}(X_i;Y)$ between each layer and the class label. Both are normalized to $[0,1]$.
We also use the compression ratio $\mathrm{CR}(X_i)$ as a proxy for intra-layer structural redundancy: larger values indicate greater compressibility, while values near $1$ indicate near-incompressibility.
Metric definitions appear in Appendix~\ref{app:empirical_metrics}.

Fig.~\ref{fig:two_fig} reveals a clear layer-wise mismatch. The protocol layers exhibit moderate inter-layer consensus (NMI $0.36$--$0.48$), while the Network and Transport layers are highly compressible ($44.7\times$ and $53.9\times$) yet retain relatively high task relevance, especially at the Network layer. In contrast, the encrypted Application layer shows weak cross-layer consensus (NMI $0.04$--$0.10$), near-random compressibility ($\mathrm{CR}\approx1$), and low task relevance. These non-perfect cross-layer dependencies further indicate private layer-specific evidence, motivating the use of flow behavior to reduce intra-layer nuisance redundancy while preserving shared consensus and complementary task-relevant cues.

Motivated by the above observations, we propose \sysname, a redundancy-aware cross-layer traffic representation framework. \sysname first applies lossless chunk compression and task-aware filtering to preserve ternary protocol semantics while reducing repetitive structure and task-irrelevant intra-layer noise. It then uses a compact flow-level representation of observable encrypted-flow behavior as a \emph{behavior anchor}, guiding the \emph{shared} components toward inter-layer consensus while preserving \emph{private} components as non-redundant layer-specific residual evidence. Finally, an information-theoretically grounded adaptive aggregation strategy fuses representations according to their sample-specific predictive contribution. Together, these theoretically grounded components produce compact and generalizable cross-layer representations for efficient training, low-latency inference, and diverse downstream tasks. We summarize our main contributions as follows:

\begin{itemize}[leftmargin=*,itemsep=2pt,
    topsep=3pt,
    parsep=0pt,
    partopsep=0pt]
    \item We formalize network traffic as \emph{cross-layer "multiview" data}, identifying redundancy as a first-class property affecting efficiency and generalization.
    \item Information-theoretic and compression-based analyses demonstrate substantial redundancy across and within layers in raw cross-layer traffic.
    \item We propose \sysname, using global flow behavior as an anchor to learn shared consensus and private non-redundant information across protocol layers.
    \item Extensive experiments show consistently strong performance across diverse traffic classification tasks, especially emerging encrypted application classification.
\end{itemize}

\section{Problem Formulation}
\label{sec:problem-statement}

\textbf{Notation.}
We study network traffic analysis under a \emph{cross-layer, multiview} abstraction. Given $N$ bidirectional flows with labels $Y\in\{1,\dots,C\}$, we use $M$ protocol-layer views $X=\{X_1,\dots,X_M\}$, where $X_i\in\mathbb R^{N\times d_{fi}}$ corresponds to a specific network layer, such as L2/L3/L4 headers and L7 payload content.\footnote{Commercial packet capture tools primarily provide access to L2--L4. In encrypted settings, L7 content is generally unavailable and is often ignored or approximated through statistical flow patterns or deep sequence embeddings of observable encrypted traffic behavior.} The views may differ in dimensionality and semantic granularity. We learn a compact fused representation $Z=h_\phi(X_1,\dots,X_M)$ that remains accurate, efficient, and robust under encryption.

\begin{figure}[t]
    \centering
    \captionsetup[subfigure]{skip=1pt}
    \begin{subfigure}[c]{0.35\linewidth}
        \centering
        \includegraphics[
            width=\linewidth,
            trim=10 8 10 8,
            clip
        ]{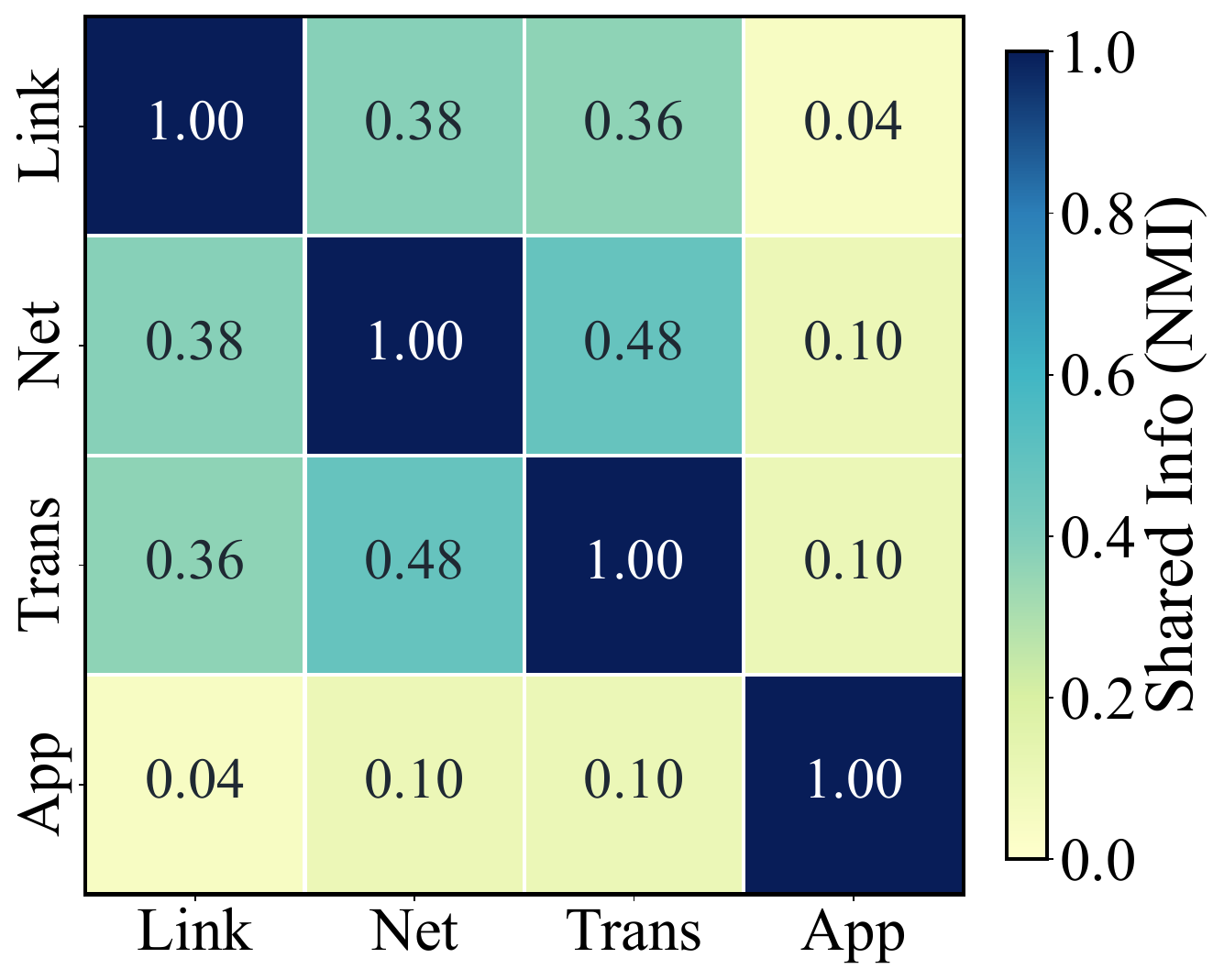}
        \caption{Inter-layer consensus}
        \label{fig:raw_cmi}
    \end{subfigure}
    \hspace{0.04\linewidth}
    \begin{subfigure}[c]{0.35\linewidth}
        \centering
        \includegraphics[
            width=\linewidth
        ]{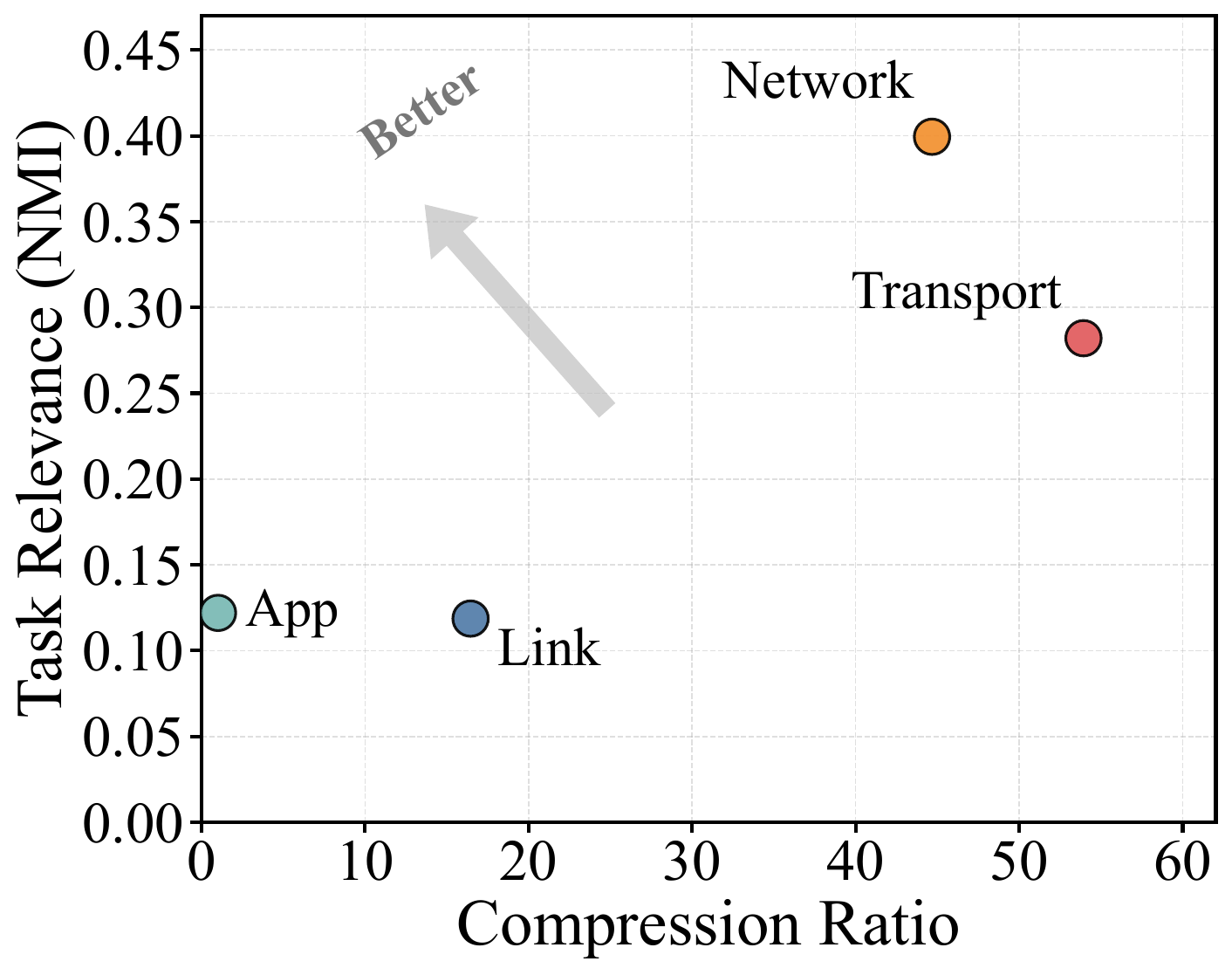}
        \caption{Task relevance and redundancy}
        \label{fig:raw_task_relevance}
    \end{subfigure}
    \caption{Layer-wise analysis on CipherSpectrum.
    (a) Inter-layer consensus measured by $\mathrm{NMI}(X_i;X_j)$.
    (b) Task relevance measured by $\mathrm{NMI}(X_i;Y)$, with intra-layer redundancy proxied by the compression ratio.
    Higher task relevance and lower compression ratio are preferred.}
    \label{fig:two_fig}
    \vspace{-0.45cm}
\end{figure}

\textbf{Multiview data non-redundancy.}
Task-relevant information may reside in both shared and unique view information. Cross-layer consensus reinforces related protocol evidence, but relying on it alone may omit complementary information available only from particular layers.

\begin{myDef}
\label{Def1}
\textit{$X_i$ is considered non-redundant with $X_j$ for $Y$ if and only if there exists $\epsilon>0$ such that $I(X_i;Y\mid X_j)>\epsilon$ or $I(X_j;Y\mid X_i)>\epsilon$.}
\end{myDef}

\textbf{Representation objectives.}
The learned representation should satisfy three requirements: \emph{layer fidelity}, preserving essential layer-specific information while reducing repetitive structure and task-irrelevant noise; \emph{redundancy-aware complementarity}, combining shared consensus $I(X_i;X_j)$ with non-redundant task information $I(X_i;Y\mid X_j)$; and \emph{compactness and efficiency}, enabling scalable training and low-latency inference with substantially smaller representations than raw-bit encodings. Our goal is to satisfy these requirements jointly under encrypted and heterogeneous traffic settings.

%% file: sections/methods.tex
\section{Methodology}

To operationalize these objectives, \sysname follows the pipeline illustrated in Fig.~\ref{fig:pipeline}. It constructs compact behavioral and protocol-header representations, filters intra-layer semantic noise, separates cross-layer consensus from private residual evidence under behavioral guidance, and adaptively aggregates their sample-specific task contributions.

\begin{figure}[t]
\centering
\includegraphics[width=\linewidth]{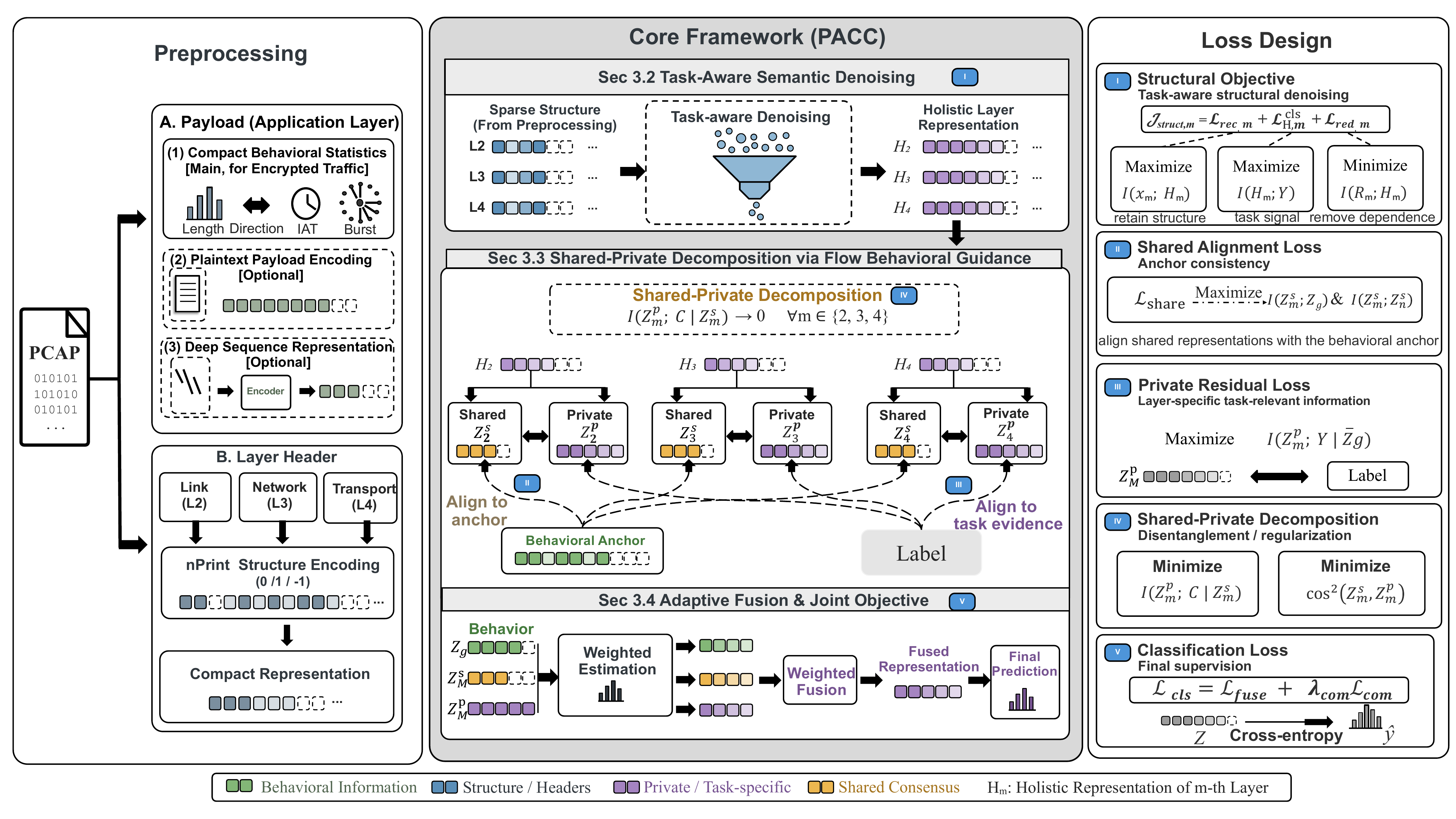}
\caption{The overall framework of \sysname. Compact behavioral features and lossless structural chunks are processed through task-aware semantic denoising, behavior-guided shared--private decomposition, and adaptive fusion. The right panel summarizes the optimization objectives.}
\label{fig:pipeline}
\vspace{-0.5cm}
\end{figure}

\subsection{Compact Behavioral and Structural Representations}

\textbf{Compact Encrypted-Flow Behavioral Representation.}
\label{4.1}
Encryption hides application content, but behavioral fingerprints, such as packet length, direction, and inter-arrival time, remain observable. We summarize these side channels as $pkt_t=(l_t,s_t,\delta_t)$, where $l_t$ is the TCP/UDP payload length, $s_t\in\{+1,-1\}$ is the direction, and $\delta_t$ is the inter-arrival time. The first observed source IP defines $+1$, and the reverse direction defines $-1$.

For each flow, we concatenate compact statistical length, direction, temporal, burst, distribution, and frequency features, and encode them as
\begin{equation}
x_{\mathrm{src}}
=
[\phi_l,\phi_d,\phi_\delta,\phi_{\mathrm{run}},
\phi_{\mathrm{hist}},\phi_{\mathrm{fft}}],
\qquad
X_P=E_P(X_{\mathrm{src}})\in\mathbb R^{N\times d_e}.
\label{eq:payload_compact_feature}
\end{equation}
Here, $E_P$ is jointly optimized with the downstream multi-layer model, and $X_P$ provides the compact flow-level behavioral representation later used to construct the consensus anchor. Detailed feature construction and preprocessing appear in Appendix~\ref{app:behavior_features}, Table~\ref{tab:payload_compact_features}. The framework can also incorporate deep packet-sequence representations like~\citet{lotfollahi2020deep} or \citet{8737507}; we use compact statistics for efficiency and to avoid an additional heavy pretrained backbone.

\textbf{Lossless Chunk Compression.}
\label{sec:segment_compression}
The link, network, and transport layers provide complementary field presence, values, and protocol-layout patterns. However, dense packet encodings such as nPrint~\citep{holland2021new} contain long repetitive absent-field and padding regions, increasing computation without necessarily adding structural information. We reduce this representation-level redundancy before learning which patterns are useful for classification.

For structural layer $m\in\mathcal M$, let $x_m\in\{-1,0,1\}^{D_m}$ denote its dense ternary sequence. The value $0$ is a valid bit state, whereas $-1$ denotes an absent or padded field; removing zero-valued entries would therefore destroy protocol semantics. Instead, run-length encoding~\citep{golomb1966run} merges consecutive identical states:
\begin{equation}
U_m=T_m(x_m)=\{(a_{m,r},\ell_{m,r},c_{m,r},g_{m,r})\}_{r=1}^{K_m}.
\end{equation}
Each chunk retains its start position $a_{m,r}$, length $\ell_{m,r}$, ternary state $c_{m,r}$, and protocol field-block identifier $g_{m,r}$. Geometry permits exact reconstruction, and field-block identity preserves protocol-layout membership rather than treating runs as unordered entries.

The distinction between absent states and valid protocol bits is essential because dense encodings are aligned to protocol layouts. A repetitive region may occupy many positions yet express only one recurring state. Recording the run and its location condenses this repetition while retaining both its span and semantic identity, reducing input length before neural representation learning.

\begin{restatable}{myProp}{propSegmentEquiv}
\label{prop:segment_equiv}
\textit{Information-Lossless Compression.
For any structural layer $m$, $U_m=T_m(x_m)$ is information-lossless with respect to $x_m$. Consequently, $I(H_m;x_m)=I(H_m;U_m)$ for any downstream representation $H_m$.}
\end{restatable}

The proof appears in Appendix~\ref{app:prop_segment}. Thus, chunks reduce physical encoding cost while retaining all information available to subsequent representations.

\subsection{Task-Aware Semantic Denoising}
\label{sec:structural_denoising}
The compressed chunks
may still include semantic task-irrelevant redundancy, such as checksums or option-related structure weakly related to the label. We therefore introduce intra-layer denoising to reduce representation dimension and suppress semantic noise while preserving useful protocol evidence.

\textbf{Chunk encoding and gating.}
Each chunk combines geometry and protocol semantics through location, length, state, and field-block embeddings. Since identical consecutive states have been merged, a lightweight neighborhood encoder captures the remaining local transitions:
\begin{equation}
\begin{aligned}
u_{m,r}&=\mathrm{LN}\big(e_a(a_{m,r})+e_\ell(\ell_{m,r})+e_c(c_{m,r})+e_g(g_{m,r})\big),\\
h_{m,r}&=F_m(u_{m,r-1},u_{m,r},u_{m,r+1}).
\end{aligned}
\end{equation}
Encoding cost is linear in the number of chunks; boundary positions and padding are handled by the token mask. Let $R_m=\{h_{m,r}\}_{r=1}^{K_m}$ be the unfiltered structural reference. A learnable gate produces the holistic representation:
\begin{equation}
\alpha_{m,r}=\sigma(f_m(h_{m,r})),\qquad
H_m=\mathrm{Pool}(\{\alpha_{m,r}h_{m,r}\}_{r=1}^{K_m}).
\end{equation}
The gate controls retained chunk contributions, and pooling aggregates their contextual evidence. We train $H_m$ to preserve ternary structure and task information while avoiding replication of $R_m$.

The unfiltered reference $R_m$ contains both discriminative patterns and semantic noise, whereas $H_m$ is a compact representation for subsequent stages. Since gating alone does not establish which information to preserve, we combine it with structural reconstruction and task supervision, balancing protocol fidelity against label-irrelevant dependence.

\begin{restatable}{myThe}{thmCrossEntropyMI}
\label{thm:cross_entropy_mi}
\textit{Cross-Entropy Surrogate for Mutual Information~\citep{cover2006elements,poole2019variational}.
Cross-entropy provides the variational lower bound $I(H;Y)\ge H(Y)-\mathcal L_{\mathrm{cls}}$. Since $H(Y)$ is fixed for a given dataset, reducing $\mathcal L_{\mathrm{cls}}$ increases this lower bound.}
\end{restatable}

\textbf{Structural-Fidelity Regularization.}
Proposition~\ref{prop:segment_equiv} establishes information equivalence before filtering. To encourage recoverable ternary semantics after gating, a decoder maps $H_m$ to logits $\mathcal D_m(H_m)\in\mathbb R^{D_m\times3}$ for states $\{-1,0,1\}$. Its reconstruction loss is
\begin{equation}
\mathcal L_{\mathrm{rec},m}=-\frac1n\sum_{i=1}^n\sum_{d=1}^{D_m}\log q_\phi(x_{m,d}^{(i)}\mid H_m^{(i)}).
\end{equation}
By Theorem~\ref{thm:cross_entropy_mi}, this maximizes a lower-bound surrogate of $I(x_m;H_m)$, grounding the filtered representation in the original structure. Ternary classification retains the distinction between valid zero-valued bits and absent fields.

Reconstruction complements the lossless transform at a different stage. Chunk compression guarantees recoverability of the structural input, while reconstruction encourages the representation to remain faithful after the gate has reweighted and pooled that input. It therefore protects layer-specific structural evidence during semantic filtering rather than determining task relevance by itself.

\textbf{Task-relevance Redundancy-Free Filtering.}
Structural fidelity alone does not remove task-irrelevant dependence. We therefore preserve label information while reducing residual dependence on the unfiltered reference.

\begin{restatable}{myDef}{defTaskSufficientFilter}
\label{def:task_sufficient_filter}
\textit{Task-relevance Redundancy-Free Filter.
$H_m$ is a task-sufficient redundancy-free filtered representation of $R_m$ , if
$I(H_m;Y)=I(R_m;Y)$ and $I(R_m;H_m\mid Y)=0$.}
\end{restatable}

\begin{restatable}{myProp}{propFilterCriterion}
\label{prop:filter_criterion}
\textit{Surrogate Criterion for Task-relevance Filtering.
Under $Y\leftrightarrow R_m\leftrightarrow H_m$, maximizing $I(H_m;Y)$ encourages the preservation of task-relevant information, while minimizing
$I(R_m;H_m)$ toward $I(H_m;Y)$ suppresses dependence on $R_m$
unexplained by the label. Together, these objectives approach the
task-sufficient redundancy-free filter in
Definition~\ref{def:task_sufficient_filter}.}
\end{restatable}

The precise bounds and proof are given in
Appendices~\ref{app:filter_criterion_details}
and~\ref{app:prop_filter_criterion}. This proposition specifies an
optimization target rather than a property already achieved by the
learned representation.

According to Theorem~\ref{thm:cross_entropy_mi}, auxiliary
classification increases a lower-bound surrogate of $I(H_m;Y)$,
preserving discriminative evidence. A CLUB-based penalty~\citep{liang2023factorized}
reduces an estimate of $I(R_m;H_m)$, discouraging dependence beyond
that needed for the task:
\begin{equation}
\begin{aligned}
\mathcal L_{H,m}^{\mathrm{cls}}
&= -\frac{1}{n}\sum_{i=1}^{n}
\log q_{H,m}\bigl(y^{(i)}\mid H_m^{(i)}\bigr),\quad
\mathcal J_{\mathrm{struct},m}
= \mathcal L_{\mathrm{rec},m}
 + \mathcal L_{H,m}^{\mathrm{cls}}
 + \mathcal L_{\mathrm{red},m}.
\end{aligned}
\end{equation}
Here, $\mathcal L_{\mathrm{red},m}=I_{\mathrm{CLUB}}(R_m;H_m)$.
The reconstruction term encourages structural fidelity, while the
classification and redundancy terms approximate the two information
objectives in Proposition~\ref{prop:filter_criterion}. The CLUB estimator is detailed in Appendix~\ref{app:denoising_details}.

\subsection{Shared-Private Decomposition via Flow Behavioral Guidance}

Filtered layers still contain heterogeneous evidence: some reflects cross-layer behavioral consistency, while other patterns arise from layer-specific protocol mechanisms. We separate these into shared consensus and private residual information so that alignment does not discard complementary task evidence.

\begin{restatable}{myDef}{defLatent}
\label{def:latent_decomp}
\textit{Flow-behavior-guided Conditional Latent Decomposition.
Two independent encoders map $H_m$ to $(Z_m^s,Z_m^p)$. The shared component captures evidence related to the cross-layer shared task factor $C$, while the private component preserves task information unexplained by the behavioral anchor. The desired split satisfies $I(Z_m^p;C\mid Z_m^s)=0$.}
\end{restatable}

Since $C$ is latent, we map the global packet-sequence embedding to $Z_g=E_g(X_P)$ and use it as an observable anchor. This compact representation summarizes overall encrypted-flow behavior, providing a common reference for structurally different protocol layers.

\begin{restatable}{myProp}{propAnchorDecomp}
\label{prop:anchor_decomp}
\textit{Anchor-Induced Shared-Private Decomposition.
Assume $H(C\mid Z_g)\le\varepsilon_g$ and
$Z_m^s\leftrightarrow C\leftrightarrow Z_g$.
Maximizing $I(Z_m^s;Z_g)$ serves as an observable objective for
capturing $I(Z_m^s;C)$ and promotes agreement among shared components.
It also constrains shared-factor information that may remain in
$Z_m^p$ after conditioning on $Z_m^s$.
Complementarily, maximizing $I(Z_m^p;Y\mid Z_g)$ preserves
task-relevant evidence beyond the behavioral anchor.}
\end{restatable}

The precise bounds and proof appear in
Appendices~\ref{app:decomposition_details}
and~\ref{app:anchor_decomp}. These conclusions depend on the stated
proxy and mediation assumptions and do not imply that the learned
branches satisfy exact conditional independence.

\textbf{Shared consensus alignment.}
We approximate the shared objective in
Proposition~\ref{prop:anchor_decomp} with the InfoNCE lower
bound~\citep{liang2023factorized,oord2018representation}:
\begin{equation}
\mathcal L_{\mathrm{share},m}
=-I_{\mathrm{lb}}(Z_m^s;Z_g),
\qquad
\mathcal L_{\mathrm{share}}
=\frac{1}{|\mathcal M|}
\sum_{m\in\mathcal M}\mathcal L_{\mathrm{share},m}.
\end{equation}
Minimizing this loss encourages each shared branch to retain
structural evidence consistent with the same global behavioral anchor.

\textbf{Private residual preservation.}
Let $\bar Z_g=\mathrm{sg}(Z_g)$ be the detached anchor.
We approximate the conditional task-information objective using
label prediction:
\begin{equation}
\mathcal L_{\mathrm{priv},m}
=-\frac{1}{n}\sum_{i=1}^{n}
\log q_{p,m}
\bigl(y^{(i)}\mid\bar Z_g^{(i)},Z_m^{p,(i)}\bigr).
\end{equation}
By the conditional form of
Theorem~\ref{thm:cross_entropy_mi}, minimizing this loss increases
a lower-bound surrogate of $I(Z_m^p;Y\mid\bar Z_g)$, encouraging
private evidence that remains useful after behavioral information
is given. The decomposition objective is
\begin{equation}
\mathcal J_{\mathrm{dec},m}
=
\mathcal L_{\mathrm{share},m}
+
\mathcal L_{\mathrm{priv},m}.
\end{equation}
Together, these objectives encourage anchor-consistent consensus
and complementary task-relevant residual evidence.

\subsection{Adaptive Fusion and Joint Objective}

Behavioral, shared, and private components provide complementary evidence whose reliability varies across samples. We therefore perform sample-specific fusion to emphasize discriminative components and suppress uncertain predictions.

Index the components by $b$, with representation $Z_b$ and classifier $q_b(Y\mid Z_b)$. Its posterior provides predictive statistics $r_b$, including confidence, prediction margin, and normalized certainty. The fusion is computed as
\begin{equation}
\hat Z_b=P_b(Z_b),\qquad
\alpha_b=
\frac{\exp(g_f(\hat Z_b,r_b))}
{\sum_{b'}\exp(g_f(\hat Z_{b'},r_{b'}))},
\qquad
Z_F=\sum_b\alpha_b\hat Z_b,
\end{equation}
where $P_b$ maps heterogeneous components into a shared latent space and $g_f$ estimates their sample-specific quality. The resulting $Z_F$ adaptively combines overall behavior, cross-layer consensus, and layer-specific residual evidence.

\textbf{Optimization.}
We jointly optimize the tractable surrogate objectives introduced above; their definitions and theoretical bounds appear in Appendix~\ref{app:denoising_details} and the corresponding proofs. Because the adaptive gate relies on component posteriors, cross-entropy supervises both the fused and component classifiers. The overall objective is
\begin{equation}
\begin{aligned}
\mathcal L
&=
\mathcal L_{\mathrm{fuse}}
+
\mathcal L_{\mathrm{com}}
+
\frac{1}{|\mathcal M|}
\sum_{m\in\mathcal M}
\left(
\lambda_{\mathrm{str}}\mathcal J_{\mathrm{struct},m}
+
\lambda_{\mathrm{dec}}\mathcal J_{\mathrm{dec},m}
\right).
\end{aligned}
\end{equation}
This objective jointly trains structural filtering, shared--private decomposition, and sample-specific fusion.

%% file: sections/experiment.tex
\section{Experiments}

We evaluate the effectiveness, efficiency, and component contributions of \sysname in reducing intra-layer redundancy, exploiting inter-layer consensus, and preserving layer-specific evidence for traffic classification.

\subsection{Experimental Setups}

\textbf{Datasets.}
We evaluate \sysname on four real-world datasets covering distinct traffic classification objectives: CipherSpectrum~\citep{11023502} for encrypted application classification, USTC-TFC~\citep{7899588} for benign/malware traffic classification, UNSW-IoT~\citep{18tmc} for IoT device classification, and ISCXTor/NonTor~\citep{inproceedings} for Tor/non-Tor recognition.
\textbf{Baselines} include flow-statistics, raw-bit/byte, and pre-trained traffic representation methods.
\textbf{Implementation Details} cover data preparation, artifact masking, feature construction, model training, and hyperparameter settings. Detailed dataset statistics, baseline descriptions, evaluation metrics, and implementation configurations are provided in Appendix~\ref{app:experimental_details}.

\textbf{Efficiency evaluation.}
Following prior studies~\citep{wang2024netmamba,wang2026netmambaplus,chehade2025energy}, we profile prepared CipherSpectrum batches on an NVIDIA RTX PRO 6000 Blackwell Server Edition GPU. We report Input/flow, Params, MACs/flow, Peak GPU, Latency/flow, Throughput, and the measured preprocessing and training time. Detailed inference and timing protocols are provided in Appendix~\ref{app:efficiency_protocol}.

\begin{table*}[!t]
\vspace{-0.5cm}
\centering
\renewcommand{\arraystretch}{0.9}
\small
\setlength{\tabcolsep}{3pt}
\captionsetup{skip=4pt}
\caption{Classification Results Across Various Domain Traffic Datasets.}
\label{tab:vertical_results}

\resizebox{0.95\textwidth}{!}{
\begin{tabular}{ll|cc|cc|cc|c}
\toprule
\multirow{2}{*}{\textbf{Dataset}} & \multirow{2}{*}{\textbf{Metric}}
& \multicolumn{2}{c|}{\textbf{Flow Statistics}}
& \multicolumn{2}{c|}{\textbf{Raw Packets}}
& \multicolumn{2}{c|}{\textbf{Pretrained Embedding}}
& \multirow{2}{*}[-1ex]{\textbf{\sysname}} \\

\cmidrule(lr){3-4} \cmidrule(lr){5-6} \cmidrule(lr){7-8}
& & \textbf{AppScanner} & \textbf{FlowPrint}
& \textbf{TFE-GNN} & \textbf{nPrint}
& \textbf{YaTC} & \textbf{NetMamba} & \\ \midrule

\multicolumn{9}{l}{\textit{\textbf{Task: Encrypted Application Classification}}} \\ \midrule
\multirow{4}{*}{CipherSpectrum}
  & ACC & 0.5356 & 0.0244 & 0.2149 & 0.3037 & \underline{0.6678} & 0.5517 & \textbf{0.9685} \\
  & PRE & \underline{0.9278} & 0.0006 & 0.2147 & 0.3933 & 0.6642 & 0.5462 & \textbf{0.9694} \\
  & REC & 0.5300 & 0.0244 & 0.2149 & 0.3037 & \underline{0.6678} & 0.5517 & \textbf{0.9685} \\
  & F1  & \underline{0.6737} & 0.0012 & 0.1642 & 0.2851 & 0.6633 & 0.5370 & \textbf{0.9685} \\ \midrule

\multicolumn{9}{l}{\textit{\textbf{Task: IoT Device Identification}}} \\ \midrule
\multirow{4}{*}{UNSW-IoT}
  & ACC & 0.8014 & 0.2438 & 0.9648 & 0.5161 & \textbf{0.9956} & \underline{0.9949} & 0.9937 \\
  & PRE & 0.8297 & 0.0114 & 0.7527 & 0.7279 & 0.8960 & \underline{0.9479} & \textbf{0.9902} \\
  & REC & 0.5296 & 0.0498 & 0.7272 & 0.5057 & 0.8983 & \underline{0.8988} & \textbf{0.9807} \\
  & F1  & 0.5998 & 0.0178 & 0.7356 & 0.5148 & 0.8964 & \underline{0.9166} & \textbf{0.9852} \\ \midrule

\multicolumn{9}{l}{\textit{\textbf{Task: Benign/Malicious Traffic Classification}}} \\ \midrule
\multirow{4}{*}{USTC-TFC}
  & ACC & 0.3240 & 0.2014 & 0.9428 & 0.9927 & \underline{0.9937} & 0.9927 & \textbf{0.9948} \\
  & PRE & 0.7794 & 0.0107 & 0.9290 & 0.9877 & \underline{0.9902} & 0.9870 & \textbf{0.9912} \\
  & REC & 0.3396 & 0.0502 & 0.9171 & \underline{0.9893} & 0.9807 & 0.9859 & \textbf{0.9902} \\
  & F1  & 0.1632& 0.0168 & 0.9207 & \underline{0.9884} & 0.9852 & 0.9861 & \textbf{0.9906} \\ \midrule

\multicolumn{9}{l}{\textit{\textbf{Task: Tor Traffic Detection}}} \\ \midrule
\multirow{4}{*}{ISCX Tor/NonTor}
  & ACC & 0.9905 & 0.9788 & 0.9969 & 0.7285 & \underline{0.9999} & \underline{0.9999} & \textbf{1.0000} \\
  & PRE & 0.6360 & 0.4894 & 0.9661 & 0.4685 & \underline{0.9999} & \underline{0.9999} & \textbf{1.0000} \\
  & REC & 0.5403 & 0.5000 & 0.9628 & 0.3966 & \underline{0.9966} & \underline{0.9966} & \textbf{1.0000} \\
  & F1  & 0.5613 & 0.4946 & 0.9644 & 0.4261& \underline{0.9982} & \underline{0.9982} & \textbf{1.0000} \\ \bottomrule

\end{tabular}
}
\vspace{-0.3cm}
\end{table*}
\begingroup
\setlength{\intextsep}{0pt}
\begin{table}[!t]
\centering
\scriptsize
\renewcommand{\arraystretch}{0.95}
\setlength{\tabcolsep}{2.5pt}
\captionsetup{skip=2pt}
\caption{Efficiency comparison on CipherSpectrum. Input size is reported as model-consumed elements for \sysname, pixels for packet-image baselines, and graph size for TFE-GNN. Lower is better for Params, MACs, Peak GPU, and Latency, while higher is better for Throughput. Pipeline time reports preprocessing and training; a dash denotes an unreported prestage cost rather than zero.}
\label{tab:cipher_cost}
\resizebox{\linewidth}{!}{
\begin{tabular}{lccccccc}
\toprule
\textbf{Method}
& \textbf{Input / flow}
& \textbf{Params}
& \textbf{MACs / flow}
& \textbf{Peak GPU}
& \textbf{Latency / flow}
& \textbf{Throughput}
& \textbf{Prestage + training} \\
\midrule
YaTC
& 1600 px
& 1.866M
& 971.9M
& 158 MB
& 0.1482 ms
& 6,749/s
& \textemdash{} $+ 1420.8s$ \\

NetMamba
& 1600 px
& 1.867M
& 683.3M
& 456 MB
& 0.0481 ms
& 20,774/s
& \textemdash{} $+ 1169.5s$ \\

TFE-GNN
& 903 nodes / 5768 edges
& 44.364M
& 157.4M
& 1297 MB
& 6.2624 ms
& 160/s
& $1991.7s$ \\
\midrule

\textbf{\sysname}
& \textbf{608 elems}
& \textbf{1.273M}
& \textbf{47.2M}
& \textbf{106 MB}
& \textbf{0.0168 ms}
& \textbf{59,660/s}
& $134.2s + 743.0s$ \\
\bottomrule
\end{tabular}
}
\end{table}
\endgroup

\makeatletter
\@startsection{subsection}{2}{\z@}%
  {-0.1pt}%
 {0.2ex}%
  {\normalsize\sc\raggedright}%
  {Results Analysis}
\makeatother
\textbf{Classification results.}
Table~\ref{tab:vertical_results} shows consistently strong performance across the four tasks, with the clearest gains on strictly masked CipherSpectrum. \sysname achieves ACC and macro F1 of $0.9685$, compared with YaTC's ACC of $0.6678$ and the strongest baseline macro F1 of $0.6737$. Artifact masking removes many shortcut features, making this setting particularly challenging for existing representations.

Flow statistics may smooth out fine-grained cross-layer agreement and rare discriminative events. FlowPrint drops sharply because its fingerprints rely heavily on address and port-related statistics, while AppScanner's behavioral features retain comparatively better performance but lack explicit cross-layer structural modeling. Raw-packet methods face encrypted-byte noise and structural redundancy; graph expansion can further increase cost. YaTC and NetMamba remain competitive, although their upstream distributions do not fully match the strictly masked setting and flattened representations may mix useful headers with encrypted-byte noise.

On UNSW-IoT, YaTC and NetMamba achieve slightly higher ACC, but \sysname obtains the highest macro PRE, REC, and F1, with F1 of $0.9852$. It also achieves the highest reported metrics on USTC-TFC and ISCX Tor/NonTor, with macro F1 of $0.9906$ and $1.0000$, respectively. These results show competitive or superior performance without expensive upstream pre-training. UNSW-IoT further suggests that useful payload-level protocol patterns can complement structural evidence when strict encryption is absent.

\textbf{Efficiency results.}
Table~\ref{tab:cipher_cost} shows that \sysname achieves low deployment cost. It uses $608$ sparse elements and $47.2$M MACs per flow, with $1.273$M parameters and $106$ MB peak GPU memory. Under the same fixed-batch, logits-only protocol, it reaches $0.0168$ ms/flow and $59{,}660$ flows/s, outperforming the compared methods in computation, memory, and inference speed. Including preprocessing, \sysname requires $877.2$ seconds, which is lower than the reported training or fine-tuning time alone for YaTC ($1420.8$ seconds), NetMamba ($1169.5$ seconds), and TFE-GNN ($1991.7$ seconds). 

Overall, the effectiveness and efficiency results support handling redundancy according to its role: compressing intra-layer structural repetition and filtering task-irrelevant noise reduce representation and computational costs, while aligning inter-layer consensus and retaining private residual cues preserve discriminative evidence. Consequently, \sysname achieves strong classification performance with compact inputs, low inference overhead, and efficient training.

\textbf{Redundancy analysis.}
We inspect learned layer representations on CipherSpectrum. Since encrypted application bytes are nearly random, their learned counterpart is the flow-level behavioral representation from Section~\ref{4.1}. Shared components follow the empirical protocol in Fig.~\ref{fig:two_fig}. For private components, we additionally measure shared-factor leakage $I(Z_m^p;Z_g\mid Z_m^s)$ and residual task evidence $I(Z_m^p;Y\mid Z_g)$; normalized estimates are defined in Appendix~\ref{app:empirical_metrics}.

Fig.~\ref{fig:two_fig1} shows stronger and more balanced shared consensus than raw views, consistent with extracting common task-relevant factors rather than merely retaining protocol co-occurrence. Shared embeddings move toward higher task relevance and lower redundancy; replacing encrypted application content with compact behavior particularly improves task relevance. Private components exhibit low to moderate anchor leakage after conditioning on $Z_m^s$, while retaining residual task evidence, especially in the Network and Transport layers. Thus, they provide complementary information rather than simply duplicating shared components.

\begin{figure}[t]
    \centering
    \begin{subfigure}[c]{0.29\linewidth}
        \includegraphics[width=\linewidth]{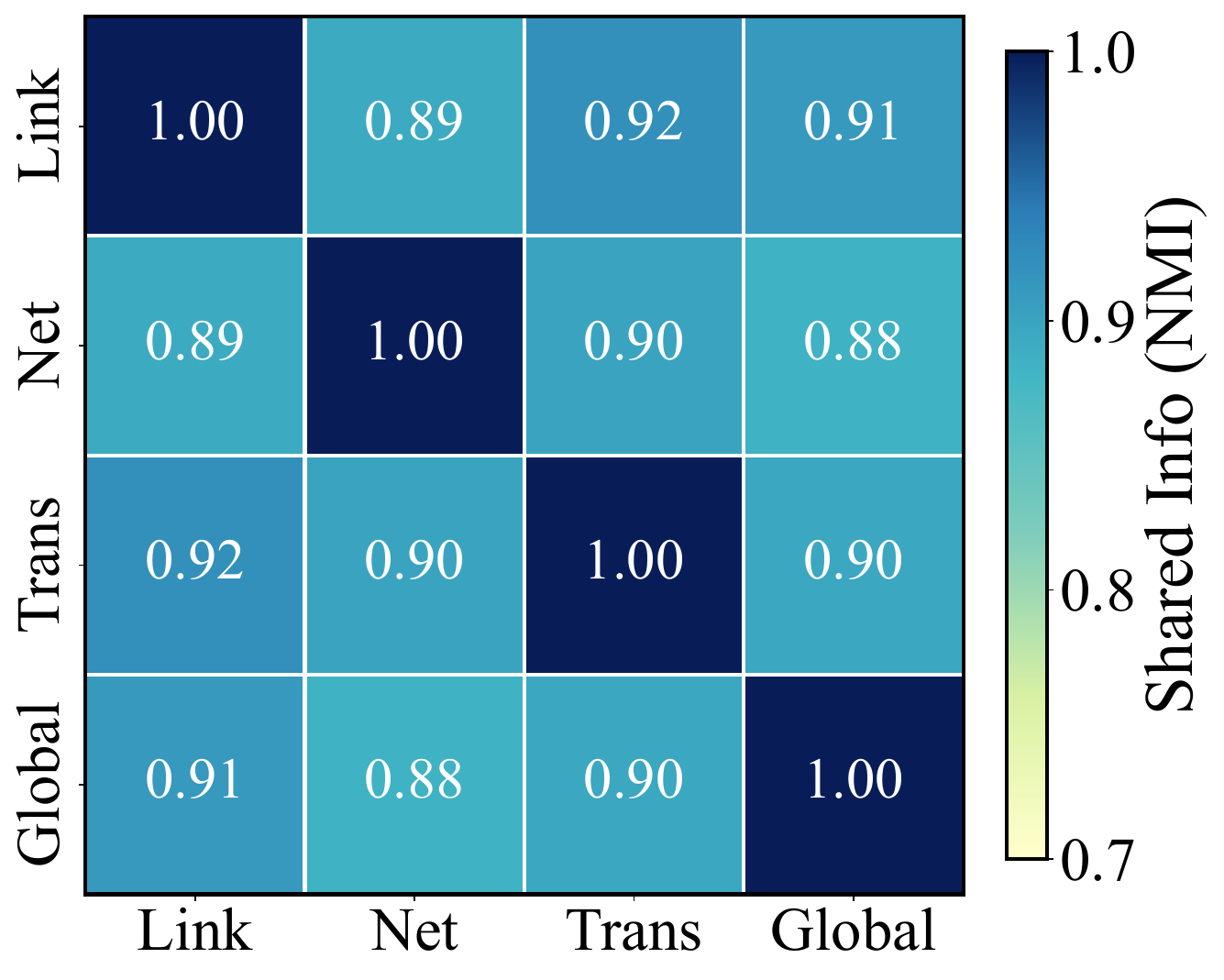}
        \caption{Shared consensus}
        \label{fig:emb_cmi}
    \end{subfigure}
    \hfill
    \begin{subfigure}[c]{0.29\linewidth}
        \includegraphics[width=\linewidth]{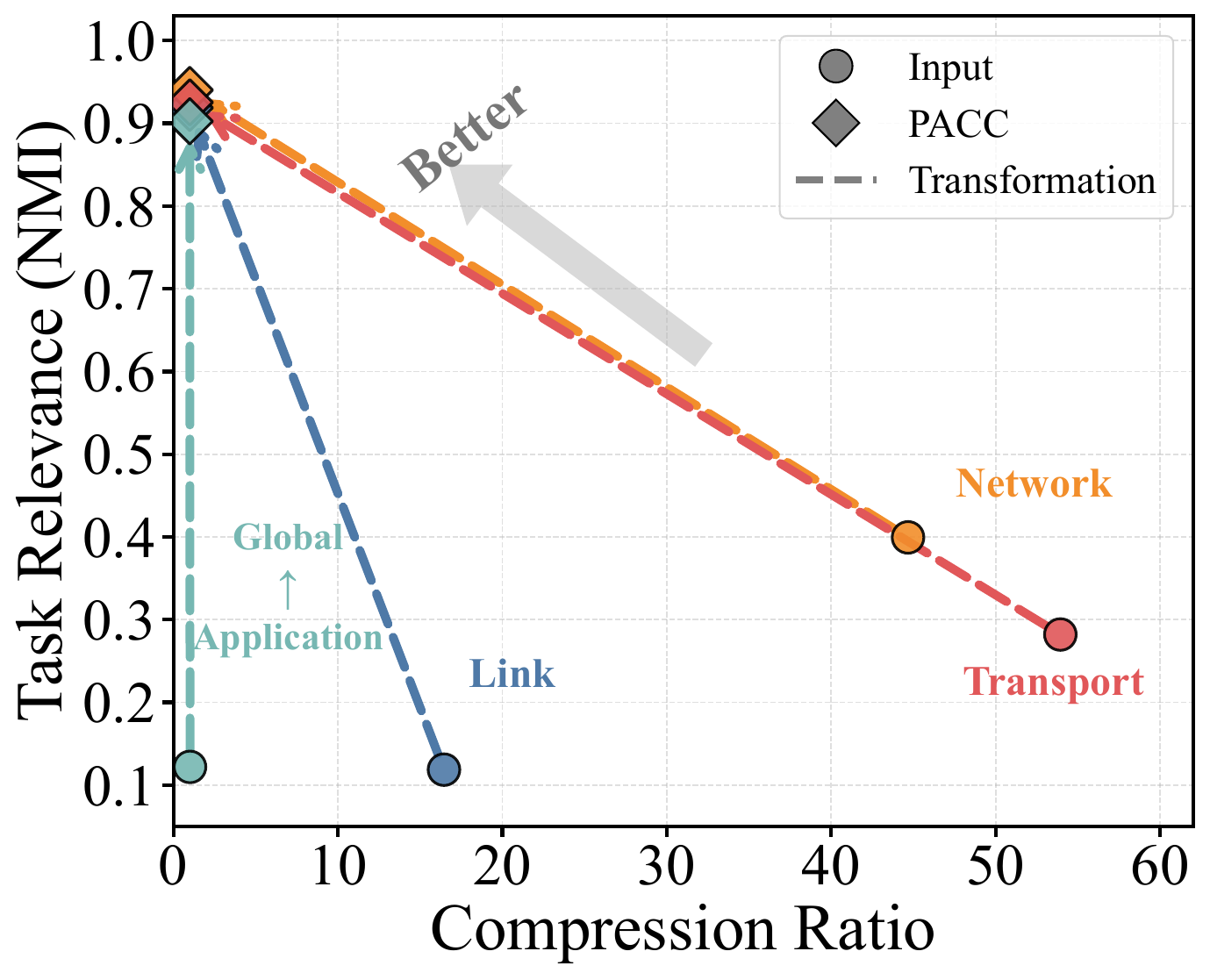}
        \caption{Shared transformation}
        \label{fig:emb Task relevance}
    \end{subfigure}
    \hfill
    \begin{subfigure}[c]{0.29\linewidth}
        \includegraphics[width=\linewidth]{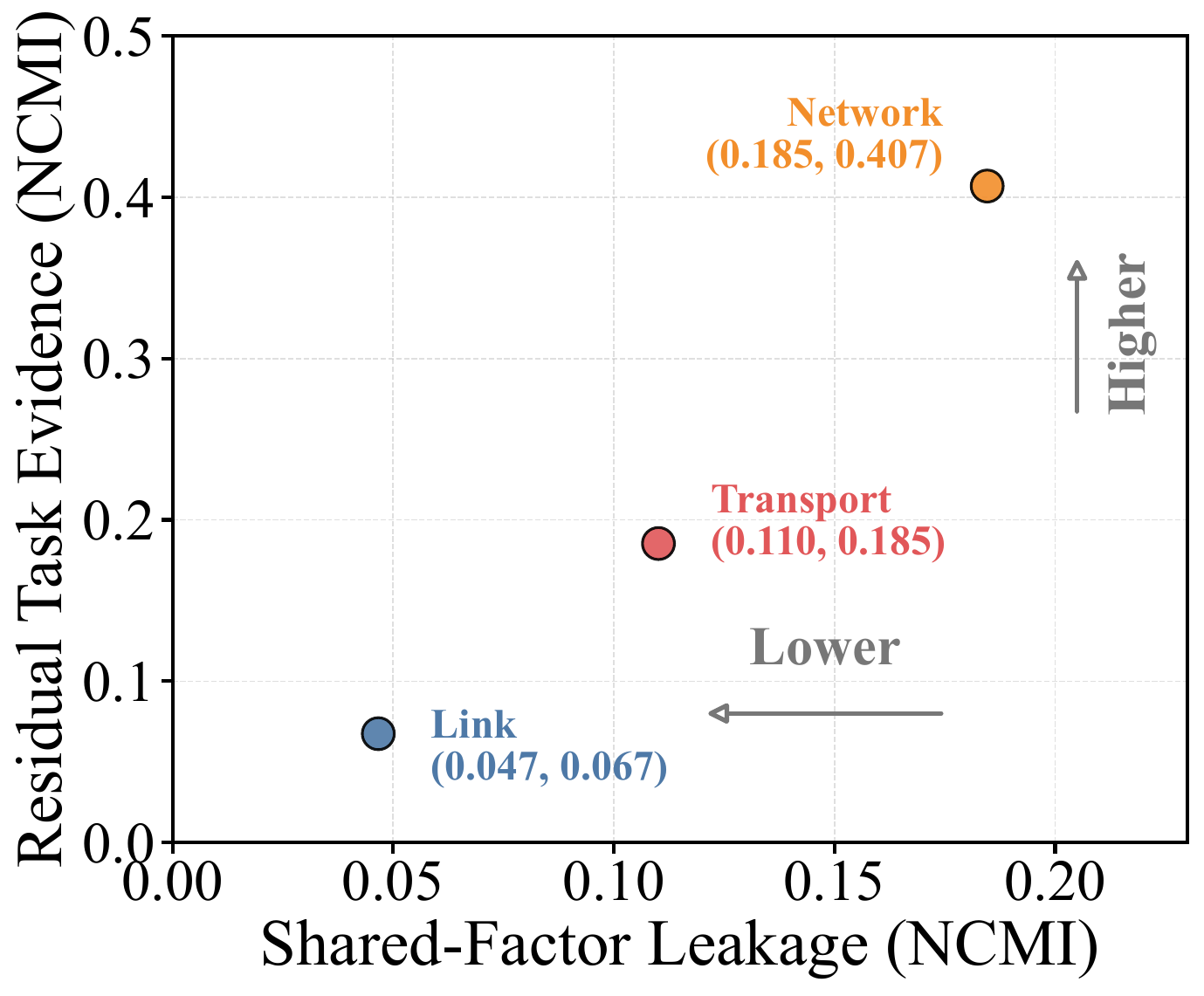}
        \caption{Private residuals}
        \label{fig:private}
    \end{subfigure}

    \caption{Layer-wise analysis of \sysname embeddings on CipherSpectrum.
(a) Shared consensus measured by $\mathrm{NMI}(Z_i^s;Z_j^s)$.
(b) Shared-component task relevance measured by $\mathrm{NMI}(Z_m^s;Y)$, with redundancy proxied by the compression ratio.
(c) Private-component shared-factor leakage measured by $I(Z_m^p;Z_g\mid Z_m^s)$ and residual task evidence measured by $I(Z_m^p;Y\mid Z_g)$.}
    \label{fig:two_fig1}
\end{figure}

\begin{figure}[t]
\centering

\begin{minipage}[t]{0.48\linewidth}
    \centering
    \begin{minipage}[c][3.45cm][c]{\linewidth}
        \centering
        \includegraphics[
            width=0.97\linewidth,
            height=3.45cm,
            keepaspectratio
        ]{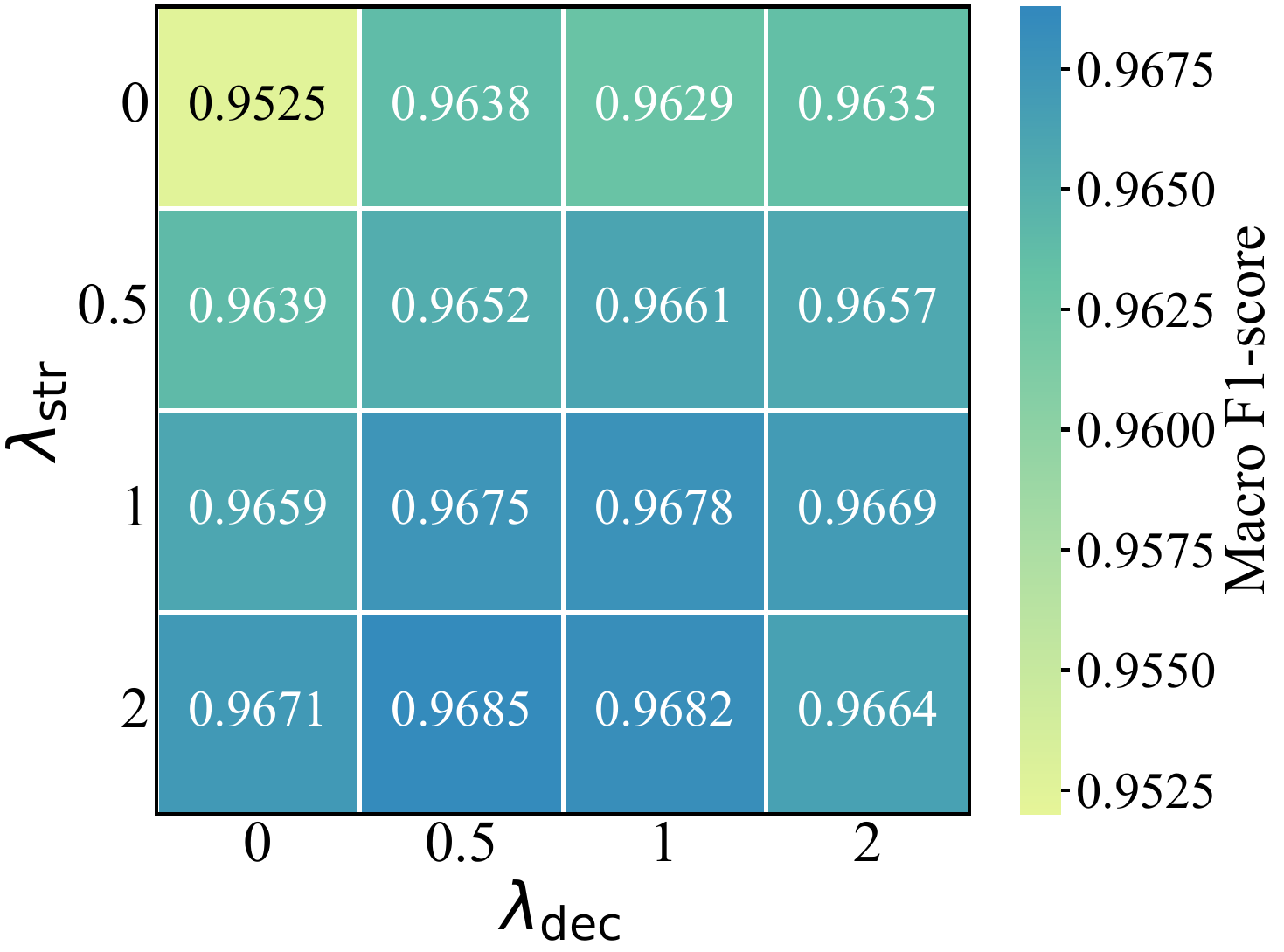}
    \end{minipage}
    \captionsetup{font=footnotesize,skip=3pt}
    \captionof{figure}{CipherSpectrum macro-F1 sensitivity to
structural/decomposition group weights.}
    \label{fig:global_weights}
\end{minipage}
\hfill
\begin{minipage}[t]{0.48\linewidth}
    \centering
    \begin{minipage}[c][2.45cm][c]{\linewidth}
        \centering
        \renewcommand{\arraystretch}{0.9}
        \setlength{\tabcolsep}{2pt}
        \resizebox{0.92\linewidth}{!}{%
            \begin{tabular}{lcccc}
            \toprule
            \textbf{Variant}
            & \textbf{Cipher}
            & \textbf{UNSW}
            & \textbf{USTC}
            & \textbf{Tor} \\
            \midrule
            w/o $H_m$
            & 0.9488 & 0.9731 & 0.9892 & 0.9982 \\
            w/o $Z_m^s$
            & 0.9037 & 0.9694 & 0.9737 & 0.9652 \\
            w/o $Z_m^p$
            & 0.9348 & 0.9532 & 0.9720 & 0.9869 \\
            w/o Fusion
            & 0.9645 & 0.9802 & 0.9774 & 0.9991 \\
            \midrule
            \sysname
            & \textbf{0.9685}
            & \textbf{0.9852}
            & \textbf{0.9906}
            & \textbf{1.0000} \\
            \bottomrule
            \end{tabular}%
        }
    \end{minipage}
    \captionsetup{font=footnotesize,skip=3pt}
    \captionof{table}{Macro-F1 architectural ablation on CipherSpectrum, UNSW-IoT, USTC-TFC, and ISCX Tor/NonTor.}
    \label{tab:module_ablation}
\end{minipage}

\vspace{-0.5cm}
\end{figure}

\tightsubsection{Ablation Study}

We evaluate the architectural contributions of \sysname by removing the filtered holistic representation $H_m$, shared component $Z_m^s$, private component $Z_m^p$, and adaptive fusion module. All variants follow the same data splits and evaluation protocol as the complete model. Objective and additional representation ablations are provided in Appendix~\ref{app:additional_ablations}. Table~\ref{tab:module_ablation} confirms contributions from all four modules.
Removing $H_m$ degrades CipherSpectrum and UNSW-IoT, supporting intra-layer
filtering before cross-layer decomposition; removing $Z_m^s$ causes the
largest CipherSpectrum drop ($0.9685$ to $0.9037$), confirming the importance
of cross-layer consensus; removing $Z_m^p$ lowers CipherSpectrum and UNSW-IoT
to $0.9348$ and $0.9532$, showing that shared information cannot replace
non-redundant layer-specific evidence; adaptive fusion provides consistent
gains.
\tightsubsection{Hyperparameter Sensitivity Analysis}
\label{sec:sensitivity_analysis}

We analyze the balance between structural filtering and shared/private decomposition on CipherSpectrum using macro F1-score. After tuning the complete objective on the validation set, we vary the group-level weights
$(\lambda_{\mathrm{str}},\lambda_{\mathrm{dec}})$ over
$\{0,0.5,1,2\}\times\{0,0.5,1,2\}$ while keeping the other settings fixed.
Fig.~\ref{fig:global_weights} shows the lowest performance when both objective groups are disabled. Moderate-to-large structural weights combined with a moderate decomposition weight perform better, suggesting that intra-layer filtering provides a foundation for organizing shared and private evidence. Excessive group weights may weaken the contribution of the main classification objective. Sensitivity to individual loss weights is analyzed in Appendix~\ref{app:sensitivity_analysis}.

%% file: sections/conclusion.tex
\tightsection{Conclusion}

{\looseness=-1
Across four traffic-classification tasks, \sysname's protocol-aware design
achieves strong performance at low per-flow cost by suppressing repetitive or
task-irrelevant intra-layer information through compact chunks and supervised
filtering, while preserving cross-layer consensus and complementary
layer-specific evidence through shared/private decomposition and adaptive
fusion. Appendix~\ref{app:additional_experiments} reports further sensitivity
and generalization analyses. Beyond the current bidirectional-flow setting,
future work will extend \sysname to session- and host-level context, adaptive
payload representations, and robustness to traffic manipulation.
\par}

%% file: sections/statement.tex
\section*{Statements}

\subsection*{AI Use Statement}
We used large language models for English-language polishing,
translation, literature search, and assistance with experimental
analysis. We also used general-purpose coding agents to assist with
writing and debugging research code and experiment scripts.
The authors reviewed AI-assisted text, code, and analyses, verified
the reported results, and take full responsibility for the scientific
claims and content of this paper.

\subsection*{Ethics Statement}
This work aims to support network security and management. Our
experiments primarily use publicly available datasets. For data
collected from an authorized network, we removed network identifiers
and other information that could identify users or devices before
analysis. Traffic classification could be misused for surveillance;
we therefore advocate its use only in authorized settings with
appropriate privacy safeguards.

%% file: sections/appendix.tex
\newpage

\section{Related Work}

\textbf{Network traffic representations.}
A network flow is a temporally ordered packet sequence exchanged between two endpoints, with each packet exposing evidence across the protocol stack.
Link, network, and transport headers provide metadata and control semantics; the application layer, often encrypted~\citep{rfc8446}, still exposes behavior through record sizes, burst patterns, and timing.
A \textbf{representation} abstracts these flows into a feature space capturing communication patterns~\citep{7899588,9127874}.
Existing methods organize this layered evidence with different trade-offs.

\textbf{Flow statistics} summarize packet counts, size distributions, inter-arrival times, and burst metrics~\citep{jiang2025jiti,liu2023amir,yin2022practical,wan2025cato,taylor2017robust,van2020flowprint}.
They are lightweight and condense repetition, but may smooth out fine-grained cross-layer agreement and rare discriminative events, making them fragile across environments as encryption and protocol evolution reduce stable plaintext information.
Neural approaches use application-layer packet-sequence features, including sizes, directions, and timing, to model encrypted traffic behavior~\citep{8737507,Rimmer2018,10.1145/3243734.3243768}.
Their performance often depends on feature granularity: coarse features may miss discriminative behavior, while detailed sequences increase computational cost and still largely overlook protocol header-field semantics across layers.

\textbf{Raw-bit/byte encodings}, including Deep Packet~\citep{lotfollahi2020deep} and nPrint-style~\citep{holland2021new,jiang2024netdiffusion} representations, directly encode headers~\citep{end2end,tfegnn2023zhang,ebsnn2022xi,meng2022packet} and sometimes payload bytes~\citep{huoh2022flow}.
Although potentially accurate, these sparsely high-dimensional sequences can contain thousands of bits/tokens per flow.
Predictable header fields, padding, and session-specific artifacts increase training difficulty and inference cost, as well as the risk of learning incidental correlations or task-irrelevant noise.

\textbf{Pre-trained embeddings} transfer knowledge from large backbones and demonstrate superior performance on downstream tasks~\citep{zhao2023yet,lin2022bert,zhou2025trafficformer,ptu2024,guthula2023netfound}.
However, many pipelines flatten protocol stacks into a single stream, entangling cross-layer signals and underutilizing protocol hierarchy.
Their dependence on upstream data distributions can also lead to degradation under downstream domain shifts, requiring costly re-pre-training or adaptation.
These limitations motivate a redundancy-aware \emph{cross-layer multiview abstraction} that combines compact layer-wise representations with shared consensus and complementary private information.

\textbf{Representation optimization in other domains.}
Multiview and multimodal representation learning offer two relevant directions.
\emph{Contrastive alignment} uses cross-modal agreement as supervision~\citep{radford2021learning,tian2020contrastive,yuan2021multimodal,mustafa2022multimodal}, tending to emphasize stable shared semantics while reducing dependence on syntactic repetition and environment-dependent noise that do not align across views.
\emph{Multimodal fusion} integrates modality-specific features into joint embeddings, preserving complementary signals and often improving over uni-modal baselines~\citep{zadeh2017tensor,tsai2019MULT,guo2025disentangling}.
These findings suggest treating protocol layers as separate but correlated modalities: learning cross-layer consensus alongside layer-specific information can help mitigate the brittleness of single-layer features.
\sysname follows this perspective while using compact global flow behavior as an anchor and explicitly distinguishing intra-layer nuisance redundancy from inter-layer consensus.

\section{Supplementary Method Details}

\subsection{Compact Packet-Sequence Feature Construction}
\label{app:behavior_features}

For each flow, we compute six groups of statistics from the observable packet attributes defined in Section~\ref{4.1}. All statistics are extracted independently per flow; raw payload bytes are not inspected. Table~\ref{tab:payload_compact_features} lists the features in each group.

\begin{table}[t]
\centering
\small
\caption{Compact packet-sequence feature construction.}
\label{tab:payload_compact_features}
\begin{tabular}{@{}p{0.18\linewidth}p{0.43\linewidth}p{0.31\linewidth}@{}}
\toprule
\textbf{View} & \textbf{Extracted statistics} & \textbf{Captured behavior} \\
\midrule
Length scale
& Packet count, total bytes, mean/std/median/90th-percentile/max packet length, and zero-payload ratio
& Traffic volume and packet-size variation. \\
Direction
& Forward/backward packet ratios, forward/backward byte ratios, and direction-switching rate
& Endpoint dominance and request-response interaction. \\
Temporal
& Summary statistics of log-compressed inter-arrival times
& Transmission pace, timing fluctuation, and delayed responses. \\
Run-level bursts
& Number of same-direction runs, average/max run length, and average/max run bytes
& Alternating exchanges and directional bursts. \\
Distribution
& Normalized histograms of packet lengths and log-compressed inter-arrival times
& Distributional shape beyond low-order moments. \\
Frequency
& Low-frequency Fourier magnitudes of signed packet-length and timing sequences
& Burst periodicity and request-response rhythm. \\
\bottomrule
\end{tabular}
\end{table}

The length-scale group summarizes the amount of traffic and variation in packet sizes. The zero-payload ratio distinguishes packets carrying no TCP/UDP payload from those carrying data. Directional packet and byte ratios characterize how traffic is distributed between endpoints, while the direction-switching rate captures how often the flow alternates between them.

Temporal statistics summarize inter-arrival times after log compression, reducing the influence of extreme timing values. Consecutive packets traveling in the same direction form a run; run counts, lengths, and bytes distinguish short alternating exchanges from longer directional bursts. Normalized length and timing histograms retain coarse distributional shape that may not be captured by summary statistics alone.

For frequency features, packet lengths are log-scaled and multiplied by packet direction before the Fourier transform so that the sequence retains bidirectional structure. We retain only leading low-frequency magnitudes of the packet-length and timing sequences. These components summarize coarse burst periodicity and request-response rhythm while suppressing high-frequency packet-level noise.

The resulting per-flow features are concatenated as specified in Eq.~\ref{eq:payload_compact_feature}. Their preparation pipeline is
\begin{equation}
\mathrm{PCAP}
\rightarrow
X_{\mathrm{src}}\in\mathbb R^{N\times d_p}
\rightarrow
E_P
\rightarrow
X_P\in\mathbb R^{N\times d_e}.
\end{equation}

\subsection{Complete Bounds for Task-relevance Filtering}
\label{app:filter_criterion_details}

Let $\underline I(H_m;Y)$ be a lower bound of $I(H_m;Y)$ and
$\overline I(R_m,H_m)$ an upper bound of $I(R_m;H_m)$.
Under the Markov relation $Y\leftrightarrow R_m\leftrightarrow H_m$,
Proposition~\ref{prop:filter_criterion} gives
\begin{equation}
0\le I(R_m;Y)-I(H_m;Y)
\le I(R_m;Y)-\underline I(H_m;Y),
\end{equation}
and
\begin{equation}
0\le I(R_m;H_m\mid Y)
\le \overline I(R_m,H_m)-\underline I(H_m;Y).
\end{equation}
The first bound controls task-information loss, while the second controls
dependence on $R_m$ unexplained by the label. Maximazing the task-information
lower bound toward $I(R_m;Y)$ and reducing the retained-dependence upper
bound toward that lower bound provide the filtering surrogate.
The proof appears in Appendix~\ref{app:prop_filter_criterion}.

\subsection{CLUB Estimation for Semantic Denoising}
\label{app:denoising_details}

To implement the redundancy-reduction objective, we use the
CLUB estimator~\citep{liang2023factorized} for the dependence between
the unfiltered representation $R_m$ and filtered representation $H_m$:
\begin{equation}
\begin{aligned}
I_{\mathrm{CLUB}}(R_m;H_m)
={}&
\mathbb E_{(R_m,H_m^+)\sim p(R_m,H_m)}
\big[f^*(R_m,H_m^+)\big]\\
&-
\mathbb E_{R_m\sim p(R_m),\,H_m^-\sim p(H_m)}
\big[f^*(R_m,H_m^-)\big].
\end{aligned}
\end{equation}
Here, $H_m^+$ is paired with $R_m$ from the same flow, whereas
$H_m^-$ is sampled from the marginal distribution of $H_m$.
The critic is $f^*(R_m,H_m)=\log q_\psi(H_m\mid R_m)$; in our
implementation, $R_m$ parameterizes a Gaussian conditional
distribution over $H_m$, and the critic evaluates its log-likelihood.
We use $I_{\mathrm{CLUB}}(R_m;H_m)$ as the redundancy penalty
$\mathcal L_{\mathrm{red},m}$ in the structural objective.

\subsection{Complete Bounds for Anchor-Induced Decomposition}
\label{app:decomposition_details}

Proposition~\ref{prop:anchor_decomp} uses
$I(Z_m^s;Z_g)$ as the observable shared objective.
Under $H(C\mid Z_g)\le\varepsilon_g$ and
$Z_m^s\leftrightarrow C\leftrightarrow Z_g$, this objective
approximates information about the latent shared factor:
\begin{equation}
I(Z_m^s;Z_g)
\le I(Z_m^s;C)
\le I(Z_m^s;Z_g)+\varepsilon_g.
\end{equation}

For two structural layers $m$ and $n$, alignment with the same
anchor gives
\begin{equation}
I(Z_m^s;Z_n^s)
\ge
I(Z_m^s;Z_g)
+
I(Z_n^s;Z_g)
-
H(Z_g).
\end{equation}
Increasing both anchor-dependence terms therefore raises a lower
bound on pairwise shared agreement. Moreover,
\begin{equation}
I(Z_m^p;C\mid Z_m^s)
\le
H(C)-I(Z_m^s;Z_g).
\end{equation}
Increasing shared-branch dependence on $Z_g$ reduces this upper
bound on additional shared-factor information in the private
component. 
For residual task information, the detached anchor
$\bar Z_g=\mathrm{sg}(Z_g)$ has the same observed values as $Z_g$
but blocks gradients through the anchor in the private objective.
Conditional cross-entropy yields
\begin{equation}
I(Z_m^p;Y\mid\bar Z_g)
\ge
H(Y\mid\bar Z_g)
-
\mathcal L_{\mathrm{priv},m}.
\end{equation}
The InfoNCE objective similarly supplies a tractable lower-bound
surrogate for $I(Z_m^s;Z_g)$.
Together, these bounds explain why shared-branch alignment targets
cross-layer consensus while conditional private supervision retains
task evidence beyond the behavioral anchor. The proof of the first
three bounds appears in Appendix~\ref{app:anchor_decomp}.

\section{Theoretical Analysis}

\subsection{Properties of Mutual Information and Representations}
\label{app:mi_properties}

In this section, we enumerate basic properties of mutual information used in
the proofs. For any random variables \(x\), \(y\), and \(z\), we have:

\noindent\((P_1)\) Non-negativity:
\begin{equation}
\label{eq:mi_nonnegativity}
I(x;y)\ge 0,\qquad I(x;y\mid z)\ge 0,\qquad H(x\mid y)\ge 0.
\end{equation}

\noindent\((P_2)\) Chain rule:
\begin{equation}
\label{eq:mi_chain_rule}
I(x,y;z)=I(y;z)+I(x;z\mid y).
\end{equation}

\noindent\((P_3)\) Entropy expansion of conditional mutual information:
\begin{equation}
\label{eq:cmi_entropy_expansion}
I(x;y\mid z)=H(y\mid z)-H(y\mid x,z).
\end{equation}

\noindent\((P_4)\) Interaction-information identity:
\begin{equation}
\label{eq:mi_interaction_identity}
I(x;y;z)
=
I(x;y)-I(x;y\mid z)
=
I(y;z)-I(y;z\mid x)
=
I(x;z)-I(x;z\mid y).
\end{equation}

\noindent\((P_5)\) Data processing inequality:
if \(x\leftrightarrow y\leftrightarrow z\), then
\begin{equation}
\label{eq:mi_data_processing}
I(x;z)\le I(x;y).
\end{equation}

\subsection{Proof of Proposition~\ref{prop:segment_equiv}}
\label{app:prop_segment}

By construction, the chunk transform \(T_m\) maps the dense ternary sequence
\(x_m\) to \(U_m\) by merging consecutive positions with identical states and
recording the start position, length, state, and semantic field-block identifier
of each chunk. Thus,
\begin{equation}
\label{eq:segment_forward_deterministic}
U_m=T_m(x_m),
\qquad
H(U_m\mid x_m)=0.
\end{equation}
Conversely, the recorded chunk information reconstructs the original dense
sequence exactly:
\begin{equation}
\label{eq:segment_inverse_deterministic}
x_m=T_m^{-1}(U_m),
\qquad
H(x_m\mid U_m)=0.
\end{equation}

We first expand the joint mutual information by the chain rule:
\begin{equation}
\label{eq:segment_dense_chain}
I(H_m;x_m,U_m)
\overset{(P_2)}{=}
I(H_m;x_m)+I(H_m;U_m\mid x_m).
\end{equation}
The conditional term satisfies
\begin{equation}
\label{eq:segment_forward_conditional_zero}
\begin{aligned}
0
&\overset{(P_1)}{\le}
I(H_m;U_m\mid x_m) \\
&\overset{(P_3)}{=}
H(U_m\mid x_m)-H(U_m\mid H_m,x_m) \\
&\le H(U_m\mid x_m)
\overset{\ref{eq:segment_forward_deterministic}}{=}0 .
\end{aligned}
\end{equation}
Therefore,
\begin{equation}
\label{eq:segment_forward_conditional_exact_zero}
I(H_m;U_m\mid x_m)=0.
\end{equation}
Then,
\begin{equation}
\label{eq:segment_joint_equals_dense}
\begin{aligned}
I(H_m;x_m,U_m)
&\overset{\ref{eq:segment_dense_chain}}{=}
I(H_m;x_m)+I(H_m;U_m\mid x_m) \\
&\overset{\ref{eq:segment_forward_conditional_exact_zero}}{=}
I(H_m;x_m).
\end{aligned}
\end{equation}

Similarly, using the inverse deterministic reconstruction in
Eq.~\ref{eq:segment_inverse_deterministic},
\begin{equation}
\label{eq:segment_backward_conditional_exact_zero}
I(H_m;x_m\mid U_m)=0.
\end{equation}
Applying the chain rule again,
\begin{equation}
\label{eq:segment_joint_equals_sparse}
\begin{aligned}
I(H_m;x_m,U_m)
&\overset{(P_2)}{=}
I(H_m;U_m)+I(H_m;x_m\mid U_m) \\
&\overset{\ref{eq:segment_backward_conditional_exact_zero}}{=}
I(H_m;U_m).
\end{aligned}
\end{equation}
Thus,
\begin{equation}
\label{eq:segment_information_equivalence}
I(H_m;x_m)
\overset{\ref{eq:segment_joint_equals_dense}}{=}
I(H_m;x_m,U_m)
\overset{\ref{eq:segment_joint_equals_sparse}}{=}
I(H_m;U_m).
\end{equation}
This proves that the chunk representation is sparse in length but
information-lossless with respect to the dense structural input.

\subsection{Proof of Proposition~\ref{prop:filter_criterion}}
\label{app:prop_filter_criterion}

Fix the learned filtering map. Since the filter takes the unfiltered structural representation $R_m$ as input and does not directly observe the label $Y$, its conditional distribution satisfies
\begin{equation}
\label{eq:filter_cond_ind}
p(H_m\mid R_m,Y)=p(H_m\mid R_m).
\end{equation}
Therefore,
\begin{equation}
\label{eq:filter_markov_zero}
\begin{aligned}
I(H_m;Y\mid R_m)
&\overset{(P_3)}{=}
H(H_m\mid R_m)-H(H_m\mid R_m,Y) \\
&\overset{\ref{eq:filter_cond_ind}}{=}
0 .
\end{aligned}
\end{equation}
Hence, the filtering architecture induces the Markov relation $Y\leftrightarrow R_m\leftrightarrow H_m$. We first derive the task-information gap. By the chain rule,
\begin{equation}
\label{eq:filter_task_ceiling_a}
\begin{aligned}
I(R_m,H_m;Y)
&\overset{(P_2)}{=}
I(R_m;Y)+I(H_m;Y\mid R_m) \\
&\overset{\ref{eq:filter_markov_zero}}{=}
I(R_m;Y).
\end{aligned}
\end{equation}
Similarly,
\begin{equation}
\label{eq:filter_task_ceiling_b}
\begin{aligned}
I(R_m,H_m;Y)
&\overset{(P_2)}{=}
I(H_m;Y)+I(R_m;Y\mid H_m) \\
&\overset{(P_1)}{\ge}
I(H_m;Y).
\end{aligned}
\end{equation}
Thus,
\begin{equation}
\label{eq:filter_task_information_order}
I(H_m;Y)\le I(R_m;Y).
\end{equation}
Since $\underline{I}(H_m;Y)$ is a lower bound of $I(H_m;Y)$, we obtain
\begin{equation}
\label{eq:filter_task_gap_bound}
0
\overset{\ref{eq:filter_task_information_order}}{\le}
I(R_m;Y)-I(H_m;Y)
\le
I(R_m;Y)-\underline{I}(H_m;Y).
\end{equation}
This proves the first bound.

Next, using the interaction-information identity in Eq.~\ref{eq:mi_interaction_identity}, we have
\begin{equation}
\label{eq:filter_interaction_identity}
I(R_m;H_m)-I(R_m;H_m\mid Y)
\overset{(P_4)}{=}
I(H_m;Y)-I(H_m;Y\mid R_m).
\end{equation}
Therefore,
\begin{equation}
\label{eq:filter_decomp}
\begin{aligned}
I(R_m;H_m)
&\overset{\ref{eq:filter_interaction_identity}}{=}
I(H_m;Y)-I(H_m;Y\mid R_m)+I(R_m;H_m\mid Y) \\
&\overset{\ref{eq:filter_markov_zero}}{=}
I(H_m;Y)+I(R_m;H_m\mid Y).
\end{aligned}
\end{equation}
Equivalently,
\begin{equation}
\label{eq:filter_residual_identity}
I(R_m;H_m\mid Y)
=
I(R_m;H_m)-I(H_m;Y).
\end{equation}
Because $\overline I(R_m,H_m)$ upper-bounds $I(R_m;H_m)$ and $\underline{I}(H_m;Y)$ lower-bounds $I(H_m;Y)$,
\begin{equation}
\label{eq:filter_residual_gap_bound}
0
\overset{(P_1)}{\le}
I(R_m;H_m\mid Y)
\overset{\ref{eq:filter_residual_identity}}{=}
I(R_m;H_m)-I(H_m;Y)
\le
\overline I(R_m,H_m)-\underline{I}(H_m;Y).
\end{equation}
This proves the second bound.

If $I(R_m;Y)-\underline{I}(H_m;Y)$ approaches zero, Eq.~\ref{eq:filter_task_gap_bound} forces
\begin{equation}
I(H_m;Y)=I(R_m;Y).
\end{equation}
If $\overline I(R_m,H_m)-\underline{I}(H_m;Y)$ approaches zero, Eq.~\ref{eq:filter_residual_gap_bound} forces
\begin{equation}
I(R_m;H_m\mid Y)=0.
\end{equation}
These two limiting conditions are exactly the two requirements in Definition~\ref{def:task_sufficient_filter}. Therefore, maximizing the task-information lower bound and minimizing the retained-dependence upper bound provides a tractable surrogate for approaching the task-sufficient redundancy-free filter. This proves the proposition.

\subsection{Proof of Proposition~\ref{prop:anchor_decomp}}
\label{app:anchor_decomp}

\textit{(i) Proxy to the latent shared factor.}
Under the shared-factor relation \(Z_m^s\leftrightarrow C\leftrightarrow Z_g\), the data processing inequality gives
\begin{equation}
\label{eq:anchor_proxy_lower}
I(Z_m^s;Z_g)
\overset{(P_5)}{\le}
I(Z_m^s;C).
\end{equation}
Next, by the interaction-information identity,
\begin{equation}
\label{eq:anchor_proxy_identity}
I(Z_m^s;C)-I(Z_m^s;C\mid Z_g)
\overset{(P_4)}{=}
I(Z_m^s;Z_g)-I(Z_m^s;Z_g\mid C).
\end{equation}
Since \(Z_m^s\leftrightarrow C\leftrightarrow Z_g\),
\begin{equation}
\label{eq:anchor_cond_zero}
I(Z_m^s;Z_g\mid C)=0.
\end{equation}
Therefore,
\begin{equation}
\label{eq:anchor_proxy_upper}
\begin{aligned}
I(Z_m^s;C)
&\overset{\ref{eq:anchor_proxy_identity}}{=}
I(Z_m^s;Z_g)+I(Z_m^s;C\mid Z_g)-I(Z_m^s;Z_g\mid C) \\
&\overset{\ref{eq:anchor_cond_zero}}{=}
I(Z_m^s;Z_g)+I(Z_m^s;C\mid Z_g) \\
&\overset{(P_3)}{=}
I(Z_m^s;Z_g)+H(C\mid Z_g)-H(C\mid Z_m^s,Z_g) \\
&\overset{(P_1)}{\le}
I(Z_m^s;Z_g)+H(C\mid Z_g) \\
&\le
I(Z_m^s;Z_g)+\varepsilon_g .
\end{aligned}
\end{equation}
Combining Eq.~\ref{eq:anchor_proxy_lower} and Eq.~\ref{eq:anchor_proxy_upper}, we obtain
\begin{equation}
\label{eq:anchor_proxy_bound}
I(Z_m^s;Z_g)
\le
I(Z_m^s;C)
\le
I(Z_m^s;Z_g)+\varepsilon_g .
\end{equation}
Thus, \(I(Z_m^s;Z_g)\) is an observable proxy for the latent shared-factor information \(I(Z_m^s;C)\), with approximation error bounded by \(\varepsilon_g\).

\textit{(ii) Anchor-induced pairwise shared agreement.}
For two structural modalities \(m\) and \(n\), the interaction-information identity gives
\begin{equation}
\label{eq:anchor_pairwise_interaction}
\begin{aligned}
I(Z_m^s;Z_n^s;Z_g)
&\overset{(P_4)}{=}
I(Z_m^s;Z_n^s)-I(Z_m^s;Z_n^s\mid Z_g),\\
I(Z_m^s;Z_n^s;Z_g)
&\overset{(P_4)}{=}
I(Z_n^s;Z_g)-I(Z_n^s;Z_g\mid Z_m^s).
\end{aligned}
\end{equation}
By the chain rule,
\begin{equation}
\label{eq:anchor_pairwise_chain}
\begin{aligned}
I(Z_m^s,Z_n^s;Z_g)
&\overset{(P_2)}{=}
I(Z_m^s;Z_g)+I(Z_n^s;Z_g\mid Z_m^s) \\
&\overset{\ref{eq:anchor_pairwise_interaction}}{=}
I(Z_m^s;Z_g)+I(Z_n^s;Z_g)-I(Z_m^s;Z_n^s;Z_g).
\end{aligned}
\end{equation}
Combining Eq.~\ref{eq:anchor_pairwise_interaction} and Eq.~\ref{eq:anchor_pairwise_chain}, we have
\begin{equation}
\label{eq:anchor_pairwise_identity}
\begin{aligned}
I(Z_m^s;Z_n^s)
=
&I(Z_m^s;Z_n^s\mid Z_g)
+
I(Z_m^s;Z_g) \\
&+
I(Z_n^s;Z_g)
-
I(Z_m^s,Z_n^s;Z_g).
\end{aligned}
\end{equation}
Then,
\begin{equation}
\label{eq:anchor_pairwise_bound}
\begin{aligned}
I(Z_m^s;Z_n^s)
&\overset{\ref{eq:anchor_pairwise_identity}}{=}
I(Z_m^s;Z_n^s\mid Z_g)
+
I(Z_m^s;Z_g)
+
I(Z_n^s;Z_g)
-
I(Z_m^s,Z_n^s;Z_g) \\
&\overset{(P_1)}{\ge}
I(Z_m^s;Z_g)+I(Z_n^s;Z_g)-I(Z_m^s,Z_n^s;Z_g) \\
&\overset{(P_3)}{=}
I(Z_m^s;Z_g)+I(Z_n^s;Z_g)-H(Z_g)+H(Z_g\mid Z_m^s,Z_n^s) \\
&\overset{(P_1)}{\ge}
I(Z_m^s;Z_g)+I(Z_n^s;Z_g)-H(Z_g).
\end{aligned}
\end{equation}
Therefore, optimizing each structural shared branch toward the same anchor \(Z_g\) increases a lower bound of the pairwise shared agreement between \(Z_m^s\) and \(Z_n^s\).

\textit{(iii) Connection to the decomposition target.}
Definition~\ref{def:latent_decomp} requires the private branch to contain no additional shared-factor information once the shared branch is given. Since \(I(Z_m^p;C\mid Z_m^s)\) is not directly optimized, we use its upper bound:
\begin{equation}
\label{eq:private_shared_leakage_bound}
\begin{aligned}
I(Z_m^p;C\mid Z_m^s)
&\overset{(P_3)}{=}
H(C\mid Z_m^s)-H(C\mid Z_m^p,Z_m^s) \\
&\overset{(P_1)}{\le}
H(C\mid Z_m^s) \\
&=
H(C)-I(Z_m^s;C) \\
&\overset{\ref{eq:anchor_proxy_lower}}{\le}
H(C)-I(Z_m^s;Z_g).
\end{aligned}
\end{equation}
Thus, increasing the observable anchor dependence \(I(Z_m^s;Z_g)\) reduces an upper bound of the shared-factor information that can remain in \(Z_m^p\) after conditioning on \(Z_m^s\). Together with Eq.~\ref{eq:anchor_proxy_bound}, this shows that the anchor objective both approximates shared-factor extraction and supports the decomposition condition in Definition~\ref{def:latent_decomp}. This proves the proposition.

\section{Empirical Metrics for Redundancy Analysis}

\label{app:empirical_metrics}

For each layer representation \(X_i\), we first obtain a discrete assignment
\(A_i=q_K(X_i)\) by \(K\)-means vector quantization with \(K=32\), where \(q_K\)
denotes the quantization map. For the \(n\)-th sample, let
\(A_i^{(n)}\in\{1,\ldots,K\}\) denote its quantized assignment. We use the same
quantization procedure for all layers.

The mutual information terms are estimated from empirical histograms. For two
layers \(X_i\) and \(X_j\), we define
\begin{equation}
\hat p_{ij}(a,b)
=
\frac{1}{N}
\sum_{n=1}^{N}
\mathbf{1}\!\left[A_i^{(n)}=a,\ A_j^{(n)}=b\right],
\end{equation}
with marginals \(\hat p_i(a)=\sum_b\hat p_{ij}(a,b)\) and
\(\hat p_j(b)=\sum_a\hat p_{ij}(a,b)\). The plug-in mutual information estimator is
\begin{equation}
I(A_i;A_j)
=
\sum_{a=1}^{K}
\sum_{b=1}^{K}
\hat p_{ij}(a,b)
\log
\frac{\hat p_{ij}(a,b)}
{\hat p_i(a)\hat p_j(b)},
\end{equation}
where zero-probability terms are omitted. The entropy is computed as
\(H(A_i)=-\sum_a \hat p_i(a)\log \hat p_i(a)\).

The cross-layer consensus between two layers \(X_i\) and \(X_j\) is measured by
normalized mutual information:
\begin{equation}
\mathrm{NMI}(X_i;X_j)
=
\frac{2I(A_i;A_j)}{H(A_i)+H(A_j)}.
\end{equation}
Similarly, the task relevance of layer \(X_i\) is measured by the normalized
mutual information between its discrete assignment and the class label \(Y\):
\begin{equation}
\mathrm{NMI}(X_i;Y)
=
\frac{2I(A_i;Y)}{H(A_i)+H(Y)}.
\end{equation}
Both quantities lie in \([0,1]\), where larger values indicate stronger shared
information.

To estimate the compressible structural part of intra-layer redundancy, we
serialize each sample representation \(x_i^{(n)}\) into bytes and compute its
generic compression ratio:
\begin{equation}
\mathrm{CR}(X_i)
=
\frac{\sum_n |\mathrm{bytes}(x_i^{(n)})|}
{\sum_n |\mathrm{zlib}(\mathrm{bytes}(x_i^{(n)}))|}.
\end{equation}
Here, \(|\cdot|\) denotes the byte length, and \(\mathrm{zlib}(\cdot)\) denotes generic lossless compression using the zlib algorithm. A larger \(\mathrm{CR}(X_i)\) indicates that the layer representation contains more compressible structural repetition, while \(\mathrm{CR}(X_i)\approx1\) indicates an almost incompressible representation.

For the private-branch analysis in Fig.~\ref{fig:private}, we use the same
\(K\)-means quantization protocol. For each structural modality \(m\), let
\(Z_m^s\) and \(Z_m^p\) denote the learned shared and private representations,
and let \(Z_g\) denote the global flow-behavior anchor. We obtain discrete
assignments
\[
S_m=q_K(Z_m^s),\qquad P_m=q_K(Z_m^p),\qquad G=q_K(Z_g).
\]
The additional shared-factor information retained by the private branch is
estimated by
\begin{equation}
I(P_m;G\mid S_m)
=
\sum_{p,g,s}
\hat p(p,g,s)
\log
\frac{\hat p(p,g,s)\hat p(s)}
{\hat p(p,s)\hat p(g,s)} ,
\end{equation}
and reported as the normalized conditional mutual information
\begin{equation}
\mathrm{NCMI}(Z_m^p;Z_g\mid Z_m^s)
=
\frac{2I(P_m;G\mid S_m)}
{H(P_m\mid S_m)+H(G\mid S_m)} .
\end{equation}
A smaller value indicates less residual shared-factor leakage in the private
branch after the shared representation is given.

Similarly, the residual task evidence preserved by the private branch is
estimated by
\begin{equation}
I(P_m;Y\mid G)
=
\sum_{p,y,g}
\hat p(p,y,g)
\log
\frac{\hat p(p,y,g)\hat p(g)}
{\hat p(p,g)\hat p(y,g)} ,
\end{equation}
and reported as
\begin{equation}
\mathrm{NCMI}(Z_m^p;Y\mid Z_g)
=
\frac{2I(P_m;Y\mid G)}
{H(P_m\mid G)+H(Y\mid G)} .
\end{equation}
A larger value indicates that the private branch retains more task-relevant
residual evidence after the global flow-behavior anchor is given.

\section{Experimental Settings and Protocols}
\label{app:experimental_details}

\subsection{Datasets, Baselines, and Evaluation Metrics}

\textbf{Datasets.}
CipherSpectrum~\citep{11023502} is a high-strength encrypted traffic benchmark for application classification under strict artifact masking. USTC-TFC~\citep{7899588} contains benign and malware traffic. UNSW-IoT~\citep{18tmc} supports IoT device classification and includes mixed application protocols that may retain useful packet-level payload patterns. ISCXTor/NonTor~\citep{inproceedings} supports Tor/non-Tor recognition, where compact global traffic behavior can provide discriminative evidence.

Table~\ref{tab:dataset_stats_main} reports the numbers of classes and flows, the final train/validation/test splits, and the corresponding data sizes. Flow-PCAP size is measured after organizing packets into bidirectional-flow samples and splitting them. Processed input size denotes the preprocessed representation size after reducing structurally repetitive information.

\begin{table}[ht]
\centering
\small
\renewcommand{\arraystretch}{1.25}
\setlength{\tabcolsep}{4.5pt}
\captionsetup{skip=4pt}
\caption{Dataset statistics under the final evaluation protocol.}
\label{tab:dataset_stats_main}
\resizebox{\linewidth}{!}{%
\begin{tabular}{lcccccc}
\toprule
\textbf{Dataset}
& \textbf{Task}
& \textbf{Classes}
& \textbf{Flow Samples}
& \textbf{Train/Val/Test}
& \textbf{Flow-PCAP Size}
& \textbf{Processed Input Size} \\
\midrule
CipherSpectrum
& Encrypted Application
& 41
& 41,000
& 32,800 / 4,100 / 4,100
& 3.58 GB
& 224.39 MB \\
USTC-TFC
& Benign/Malware Traffic
& 20
& 444,071
& 355,178 / 44,440 / 44,453
& 3.65 GB
& 1.60 GB \\
UNSW-IoT
& IoT Device
& 22
& 41,176
& 32,980 / 4,098 / 4,098
& 671.14 MB
& 100.63 MB \\
ISCXTor/NonTor
& Tor/Non-Tor
& 2
& 68,258
& 54,603 / 6,827 / 6,828
& 21.59 GB
& 178.44 MB \\
\bottomrule
\end{tabular}%
}
\end{table}

\textbf{Baselines.}
We compare three representation paradigms. Flow-statistics methods include AppScanner~\citep{taylor2017robust} and FlowPrint~\citep{van2020flowprint}. Raw-bit/byte methods include TFE-GNN~\citep{tfegnn2023zhang} and nPrint~\citep{holland2021new}. For nPrint, we reproduce its feature encoding and feed the resulting representations to an MLP classifier. Pre-trained methods include YaTC~\citep{zhao2023yet} and NetMamba~\citep{wang2024netmamba}; we fine-tune their open-source pre-trained weights on the downstream tasks.

\textbf{Classification metrics.}
We report Accuracy (ACC) and macro-averaged Precision (PRE), Recall (REC), and F1-score (F1). All reported classification metrics are computed on the test split.

\subsection{Implementation Details}

\textbf{Data preparation.}
Each raw PCAP is partitioned into bidirectional flows. We standardize the packet-level input length through truncation or zero-padding and use a unified 8:1:1 training/validation/test split. To reduce shortcut learning from network configurations or temporal session states, we apply artifact masking across layers following the protocol associated with CipherSpectrum~\citep{11023502}.

\textbf{Feature construction.}
We construct standardized bit-level structural views from Link (L2), Network (L3), and Transport (L4) protocol headers using nPrint. For strongly encrypted or anonymized traffic such as CipherSpectrum, the payload view uses compact behavioral statistics instead of raw payload bytes. The raw-view analysis in Fig.~\ref{fig:two_fig} motivates this choice: encrypted application content offers limited directly accessible plaintext semantics, leaving observable behavior as an important source of evidence. For non-strictly encrypted traffic, sparse payload-bit encodings can retain useful protocol-level evidence. The structural representations and redundancy-aware learning pipeline remain unchanged while the payload encoder adapts to the traffic setting.

\textbf{Training configuration.}
We train \sysname with AdamW. Across the four datasets, the learning rate is selected from $2.0\times10^{-4}$ to $1.0\times10^{-3}$, with weight decay $1.0\times10^{-4}$. Batch sizes range from $16$ to $64$ according to dataset scale and memory footprint, and training lasts $30$--$160$ epochs. The maximum numbers of structural chunks are capped at $32$, $192$, and $128$ for the link, network, and transport layers, respectively. The hidden dimension is set to $96$ or $128$.

\subsection{Efficiency Evaluation Protocol and Metrics}
\label{app:efficiency_protocol}

\textbf{Inference profiling.}
Following prior traffic-model efficiency studies~\citep{wang2024netmamba,wang2026netmambaplus,chehade2025energy}, we profile prepared CipherSpectrum inference batches on an NVIDIA RTX PRO 6000 Blackwell Server Edition GPU. All methods use batch size $64$, $80$ warm-up iterations, and $600$ timed repeats. Latency, throughput, parameter count, MACs, and peak GPU memory refer to the deployment inference graph that produces final prediction logits.

Inference profiling excludes data loading, disk I/O, CPU preprocessing, loss computation, back-propagation, optimizer updates, logging, and auxiliary modules used solely for training-time regularization. Pre-trained baselines are profiled after fine-tuning on the same strict split; their original pre-training cost is excluded.

\textbf{Efficiency metrics.}
\textit{Input/flow} denotes the logical input consumed by each model: sparse elements for \sysname, pixels for packet-image baselines, and nodes/edges for TFE-GNN. \textit{Params} counts trainable parameters in the deployment inference graph. \textit{MACs/flow} measures per-flow multiply-accumulate operations, while \textit{Peak GPU} records peak GPU memory during profiling. \textit{Latency/flow} is the average logits-only forward time divided by batch size, and \textit{Throughput} is the number of flows processed per second.

\textbf{Pipeline timing.}
We measure \sysname from split CipherSpectrum flow PCAPs through preprocessing and training. Preprocessing all $41{,}000$ flows with eight CPU workers takes $134.2$ seconds, and training takes $743.0$ seconds. YaTC, NetMamba, and TFE-GNN are trained on the same GPU with batch size $64$ and learning rate $10^{-3}$; their training or fine-tuning times through the selected best checkpoint are $1420.8$, $1169.5$, and $1991.7$ seconds, respectively.

Baseline preprocessing and pre-training costs are unreported. Thus, the dash in Table~\ref{tab:cipher_cost} indicates missing prestage time, and the baseline entries should not be interpreted as complete pipeline totals.

\section{Supplementary Experiments}
\label{app:additional_experiments}

\subsection{Additional Ablation Studies}
\label{app:additional_ablations}

\textbf{Objective Ablation.}
We remove one training objective at a time while retaining the remaining
architecture and other loss objectives. All variants follow the data splits and
evaluation protocol of the main experiments. Table~\ref{tab:objective_ablation}
reports macro F1-scores across the four datasets.

Removing $\mathcal L_{H,m}^{\mathrm{cls}}$ noticeably degrades
performance, indicating that the filtered holistic representation
should remain task-discriminative. Removing $\mathcal L_{\mathrm{rec},m}$
also causes clear degradation, showing the importance of preserving
recoverable protocol structure during filtering. The effect of
$\mathcal L_{\mathrm{red},m}$ is milder, but its removal still slightly
reduces performance, suggesting that redundancy suppression acts as
a useful regularizer rather than the dominant training signal.
Removing $\mathcal L_{\mathrm{share}}$ weakens shared representation
alignment, while removing $\mathcal L_{\mathrm{priv},m}$ reduces the
private representation's ability to retain residual task evidence.
Together, these results support the complementary roles of
intra-layer filtering, cross-layer consensus, and private residual
preservation in PACC.

\begin{table*}[t]
\centering
\scriptsize
\renewcommand{\arraystretch}{0.9}
\setlength{\tabcolsep}{3pt}
\caption{Objective ablation across datasets. Values are macro F1-scores.}
\label{tab:objective_ablation}
\resizebox{0.74\textwidth}{!}{%
\begin{tabular}{lcccc}
\toprule
\textbf{Variant}
& \textbf{CipherSpectrum}
& \textbf{UNSW-IoT}
& \textbf{USTC-TFC}
& \textbf{ISCX Tor/NonTor} \\
\midrule
w/o $\mathcal L_{H,m}^{\mathrm{cls}}$
& 0.9425 & 0.9826 & 0.9892 & 0.9964 \\
w/o $\mathcal L_{\mathrm{rec},m}$
& 0.9568 & 0.9779 & 0.9887 & 0.9911 \\
w/o $\mathcal L_{\mathrm{red},m}$
& 0.9674 & 0.9849 & 0.9901 & 0.9998 \\
w/o $\mathcal L_{\mathrm{share}}$
& 0.9669 & 0.9750 & 0.9895 & 0.9988 \\
w/o $\mathcal L_{\mathrm{priv},m}$
& 0.9658 & 0.9837 & 0.9849 & 0.9986 \\
\midrule
PACC
& \textbf{0.9685}
& \textbf{0.9852}
& \textbf{0.9906}
& \textbf{1.0000} \\
\bottomrule
\end{tabular}%
}
\end{table*}

\medskip
\noindent\textbf{Representation and Fusion Ablation.}

We compare behavioral-only, protocol-layer-only, and combined
representations on CipherSpectrum to examine the contributions of
input sources, learned representations, and fusion utility. The
experiment uses 32,800 training and 4,100 test flows across 41
classes; the validation split is not used. Behavioral MLP and
Behavioral LightGBM classify the same 256-dimensional behavioral representations input, while L2--L4 only concatenates the nPrint-based structural inputs from the Link, Network, and Transport layers without behavioral input. Raw features concatenation directly combines the
Behavioral, Link, Network, and Transport inputs. Representations
concatenation instead combines PACC's learned branch representations
without adaptive fusion, corresponding to the w/o Fusion design in
Table~\ref{tab:module_ablation}. PACC retains all learned representations
and adaptive fusion.

Neural variants are trained for 30 epochs with AdamW, batch size 256,
learning rate $10^{-3}$, and weight decay $10^{-4}$. Behavioral MLP, L2--L4 only, and Raw features concatenation are trained
with classification cross-entropy. Representations concatenation and PACC
use the full training objectives, they differ
only in whether adaptive fusion is used. Behavioral LightGBM is selected by 40-trial randomized
five-fold cross-validation on the training split. We run each variant with 10 random seeds, evaluate its final checkpoint
once on the test split, and report the mean and sample standard
deviation of ACC, PRE, REC, and F1 in
Table~\ref{tab:cipherspectrum_representation}.

The table shows that behavioral-only
models outperform L2--L4 only, while combining all four raw inputs
improves F1 further. Learned shared-private representations raise F1
from $0.9007$ for Raw features concatenation to $0.9430$ for
Representations concatenation. Together, these results show the value of combining evidence across
protocol layers. The further gain from learned representations is
consistent with \sysname's design: reducing intra-layer redundancy
while retaining cross-layer consensus and layer-specific evidence.
Adaptive fusion provides an additional improvement and interpretability.

\begin{table*}[t]
\centering
\scriptsize
\renewcommand{\arraystretch}{1.05}
\setlength{\tabcolsep}{4pt}
\caption{Representation and fusion ablation on CipherSpectrum.
Values are mean $\pm$ sample standard deviation over 10 seeds.}
\label{tab:cipherspectrum_representation}
\resizebox{\textwidth}{!}{%
\begin{tabular}{lcccc}
\toprule
\textbf{Variant}
& \textbf{ACC}
& \textbf{PRE}
& \textbf{REC}
& \textbf{F1} \\
\midrule
Behavioral MLP
& $0.8833 \pm 0.0046$
& $0.8890 \pm 0.0021$
& $0.8833 \pm 0.0046$
& $0.8832 \pm 0.0042$ \\
Behavioral LightGBM
& $0.8818 \pm 0.0035$
& $0.8857 \pm 0.0037$
& $0.8818 \pm 0.0035$
& $0.8824 \pm 0.0033$ \\
L2--L4 only
& $0.6997 \pm 0.0604$
& $0.7225 \pm 0.0614$
& $0.6997 \pm 0.0604$
& $0.6919 \pm 0.0634$ \\
Raw features concatenation
& $0.9022 \pm 0.0039$
& $0.9203 \pm 0.0036$
& $0.9022 \pm 0.0039$
& $0.9007 \pm 0.0052$ \\
Representations concatenation
& $0.9430 \pm 0.0212$
& $0.9441 \pm 0.0225$
& $0.9430 \pm 0.0222$
& $0.9430 \pm 0.0215$ \\
\midrule
\textbf{PACC}
& $\mathbf{0.9571 \pm 0.0122}$
& $\mathbf{0.9581 \pm 0.0125}$
& $\mathbf{0.9571 \pm 0.0122}$
& $\mathbf{0.9570 \pm 0.0122}$ \\
\bottomrule
\end{tabular}%
}
\end{table*}

\subsection{Sensitivity to Individual Objectives}
\label{app:sensitivity_analysis}

Starting from the configuration selected on the validation set, we vary
one auxiliary loss weight at a time over
$\{0,10^{-3},10^{-2},10^{-1},1\}$ while keeping all other weights fixed.
This protocol isolates each objective's effect around the validated
configuration rather than re-selecting all hyperparameters.

Fig.~\ref{fig:individual_weights} reports the results. For structural
filtering, Fig.~\ref{fig:filter_rec} shows that both
$\lambda_{\mathrm{red}}$ and $\lambda_{\mathrm{rec}}$ improve performance
at moderate weights. Increasing the reconstruction weight from $0$ to
$10^{-2}$ is particularly beneficial, supporting preservation of
recoverable ternary protocol structure; an excessive redundancy weight
may instead suppress useful structural evidence.

For shared/private decomposition, Fig.~\ref{fig:share_priv} shows a
clear drop when $\lambda_{\mathrm{priv}}=0$, consistent with the ablation
results and supporting supervision of layer-specific residual evidence.
Moderate shared alignment helps exploit cross-layer consensus, whereas
an overly strong constraint may interfere with layer-specific cues.
Overall, the auxiliary objectives are most useful as regularizers rather
than dominating the main classification objective.

\begin{figure}[t]
    \centering
    \begin{subfigure}[c]{0.46\linewidth}
        \centering
        \includegraphics[width=\linewidth]{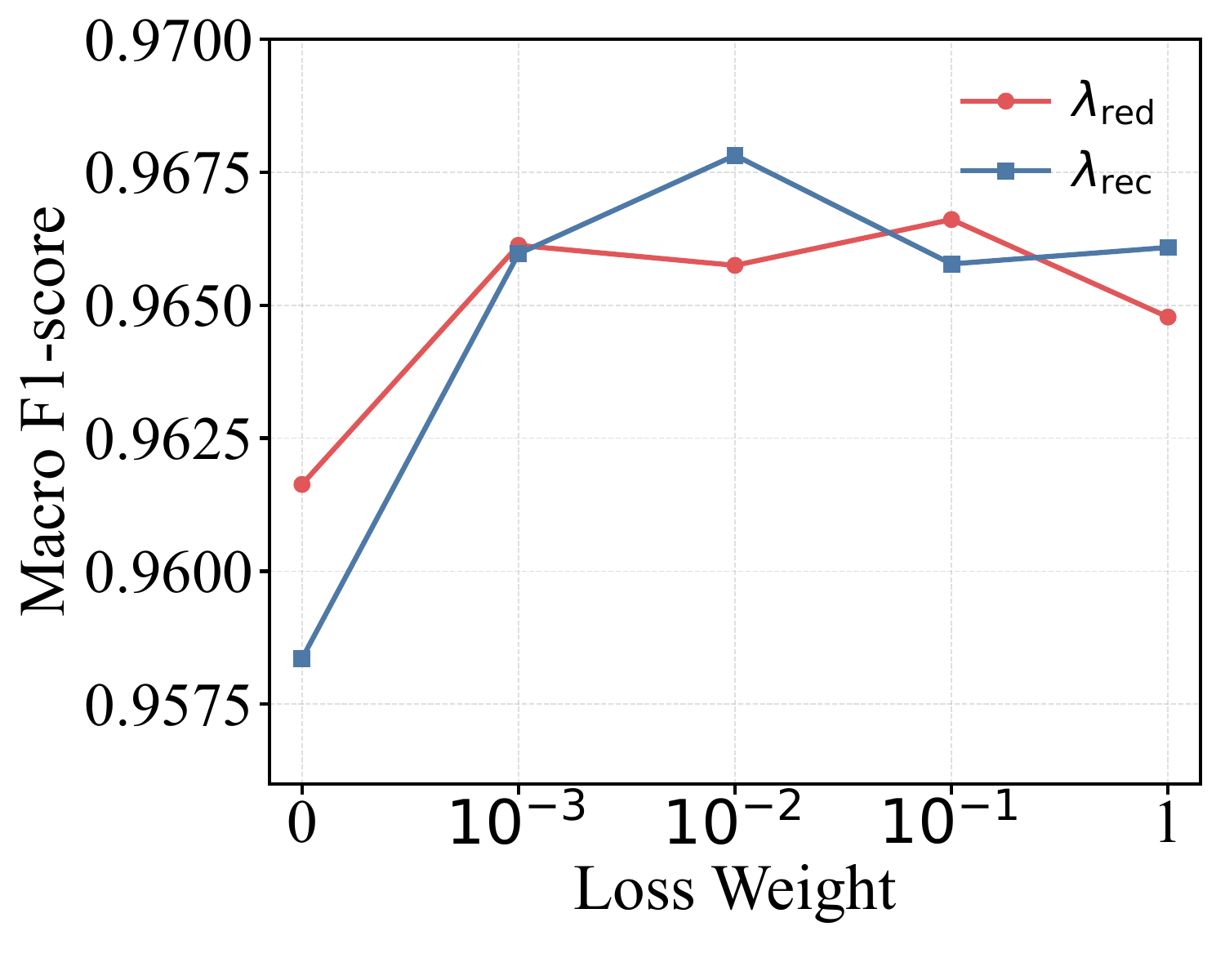}
        \caption{Filtering and reconstruction}
        \label{fig:filter_rec}
    \end{subfigure}
    \hfill
    \begin{subfigure}[c]{0.46\linewidth}
        \centering
        \includegraphics[width=\linewidth]{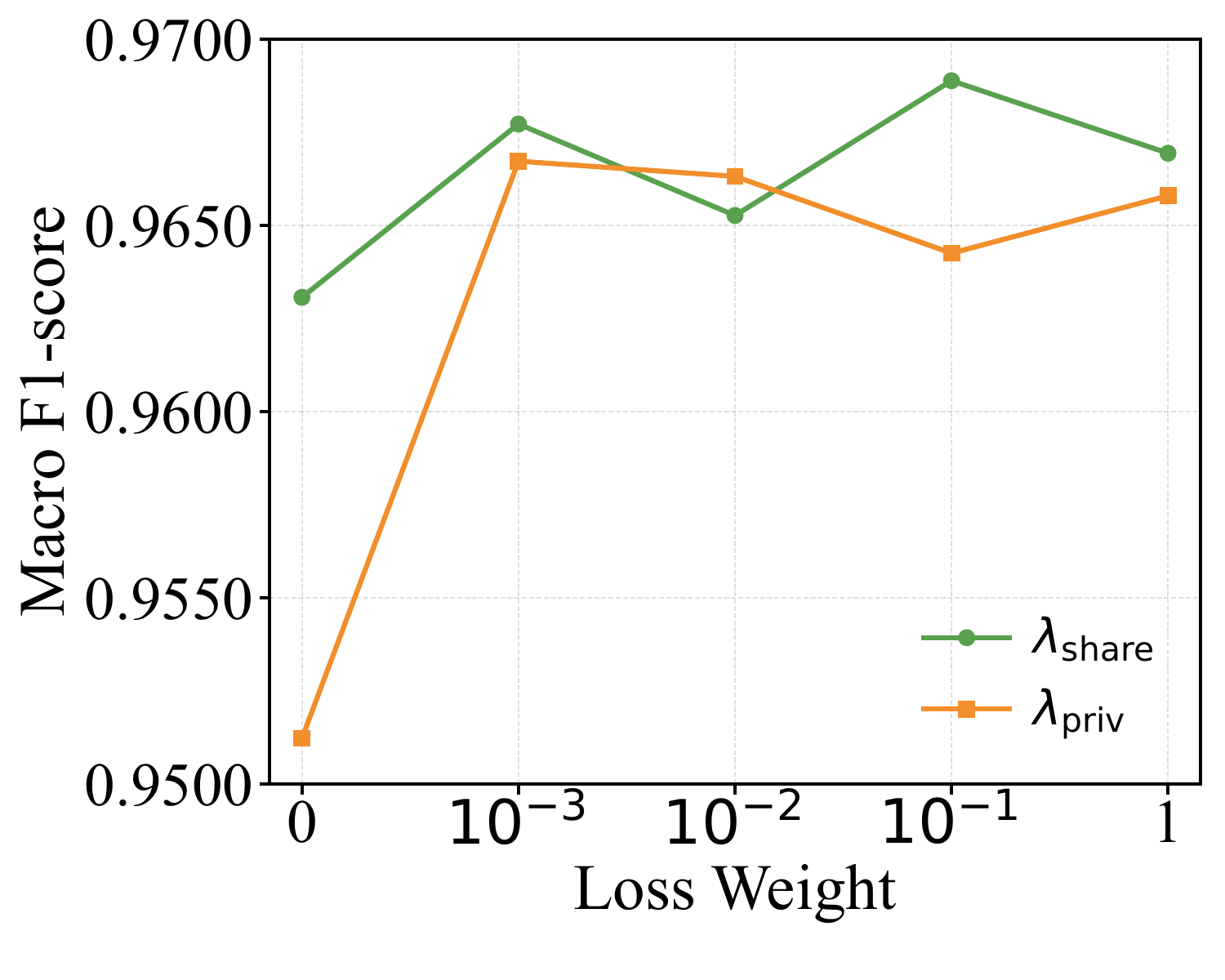}
        \caption{Shared and private objectives}
        \label{fig:share_priv}
    \end{subfigure}
    \caption{Sensitivity of macro F1-score on CipherSpectrum to individual filtering, reconstruction, shared, and private objective weights.}
    \label{fig:individual_weights}
\end{figure}

\subsection{Generalization Experiments}
\label{app:generalization_experiments}

\textbf{Tasks and metrics.}
We examine three settings beyond the main evaluation:
few-shot recognition of unseen applications on CipherSpectrum,
future-day IoT device classification on UNSW-IoT, and encrypted
application classification on a proprietary 50-class DoH dataset.
We compare PACC with YaTC and NetMamba using Accuracy (ACC) and
macro-F1.

\textbf{Few-shot setting.}
The 41 CipherSpectrum classes are partitioned into 26 base,
five development, and ten novel classes. After training on the base
classes, we freeze each encoder and L2-normalize its features.
Each 5-way episode contains $K\in\{1,5,10\}$ support examples and
15 query examples per class. Support and query flows come from the
original training and test partitions without overlap. All methods
use the same 100 novel-class episodes, nested support sets, and
query samples. A logistic-regression classifier is fitted to each
support set; its regularization parameter is selected on 30
development episodes. Table~\ref{tab:generalization_fewshot}
reports the mean over novel-class episodes.

\textbf{Temporal setting.}
For UNSW-IoT, training and validation use captures from
September 23--26, 2016, and testing uses captures from
October 8, 2016. We apply the same bidirectional sessionization
rules across dates, splitting TCP and UDP sessions after idle gaps
exceeding 120 and 60 seconds, respectively. The resulting sets
contain 38,730 training, 4,302 validation, and 8,374 test flows.
The test set covers 13 of the 22 training classes; prediction
remains over all 22 classes, while macro-F1 averages the 13
classes present in the test set. Results appear in
Table~\ref{tab:generalization_additional}.

\textbf{Additional benchmark.}
We further evaluate 50-class encrypted application classification
on a challenging DoH dataset containing strictly encrypted traffic.
We use the same 8:1:1 training/validation/test split and training
protocol as in the main experiments. Table~\ref{tab:generalization_additional}
reports ACC and macro-F1.

\begin{table*}[t]
\centering
\small
\renewcommand{\arraystretch}{1.08}

\begin{minipage}[t]{0.48\textwidth}
\vspace{0pt}
\centering
\captionsetup{font=small,skip=5pt}
\captionof{table}{Few-shot classification on CipherSpectrum (\%).}
\label{tab:generalization_fewshot}
\setlength{\tabcolsep}{5pt}
\begin{tabular*}{\linewidth}{@{\extracolsep{\fill}}clcc@{}}
\toprule
Shot & Method & ACC & Macro-F1 \\
\midrule
\multirow{3}{*}{1}
& YaTC & 50.92 & 48.15 \\
& NetMamba & 55.71 & 53.25 \\
& PACC & \textbf{60.51} & \textbf{59.22} \\
\midrule
\multirow{3}{*}{5}
& YaTC & 68.59 & 67.75 \\
& NetMamba & 72.05 & 70.54 \\
& PACC & \textbf{79.68} & \textbf{79.28} \\
\midrule
\multirow{3}{*}{10}
& YaTC & 72.84 & 72.34 \\
& NetMamba & 77.60 & 76.77 \\
& PACC & \textbf{85.55} & \textbf{85.41} \\
\bottomrule
\end{tabular*}
\end{minipage}
\hfill
\begin{minipage}[t]{0.48\textwidth}
\vspace{0pt}
\centering
\captionsetup{font=small,skip=5pt}
\captionof{table}{Temporal and 50-class DoH results (\%).}
\label{tab:generalization_additional}
\setlength{\tabcolsep}{5pt}
\begin{tabular*}{\linewidth}{@{\extracolsep{\fill}}lcc@{}}
\toprule
\multicolumn{3}{@{}l}{\textit{UNSW-IoT: temporal generalization}} \\
\midrule
Method & ACC & Macro-F1 \\
\midrule
YaTC & 77.26 & 88.50 \\
NetMamba & 80.91 & 85.45 \\
PACC & \textbf{90.03} & \textbf{88.92} \\
\midrule
\multicolumn{3}{@{}l}{\textit{50-class DoH classification}} \\
\midrule
Method & ACC & Macro-F1 \\
\midrule
YaTC & 74.97 & 78.64 \\
NetMamba & 79.72 & 80.77 \\
PACC & \textbf{83.17} & \textbf{84.52} \\
\bottomrule
\end{tabular*}
\end{minipage}
\end{table*}

PACC achieves the highest ACC and macro-F1 across all few-shot,
temporal, and 50-class DoH evaluations. These results demonstrate
strong generalization to unseen applications, future-day traffic,
and a separate proprietary dataset.